\documentclass[10pt,english]{article}
\usepackage{amsfonts}
\usepackage{amssymb,amsthm,amsmath}
\usepackage{graphicx}
\usepackage{graphics}
\usepackage{subfigure}
\usepackage{url}

\newtheorem{theorem}{Theorem}[section]

\newtheorem{lemma}[theorem]{Lemma}
\newtheorem{proposition}[theorem]{Proposition}
\newtheorem{remark}[theorem]{Remark}

\begin{document}

\title{Valuing Exchange Options Under an Ornstein-Uhlenbeck Covariance Model}
\author{Enrique Villamor, Florida International University, Pablo Olivares, Ryerson University}
\maketitle

\begin{abstract}
In this paper we study the pricing of exchange options under  a dynamic described by stochastic correlation with random jumps. In particular, we consider  a Ornstein-Uhlenbeck  covariance model in the lines of the model proposed in \cite{BN-nicolato} with Levy Background Noise Process  driven by Inverse Gaussian subordinators. We use  expansion in terms of Taylor polynomials and cubic splines to approximately compute  the price of the derivative contract. Our findings show that this approach provides an efficient way to compute the price when compared  with a Monte Carlo method while maintaining an equivalent degree of accuracy with the latter.
\end{abstract}

\date{Today}

\section{Introduction}
In this paper we study the pricing of exchange options when the underlying assets have stochastic correlation with random jumps. Specifically, we consider  an Ornstein-Uhlenbeck  covariance process  with Background Noise Levy Process(BNLP)  driven by Inverse Gaussian subordinators. In order  to calculate the price of the derivative contract we use  expansions of the conditional price in terms of  Taylor and cubic spline polynomials and compare the results with a  computationally expensive Monte Carlo method.\\
To our knowledge the problem of pricing exchange derivatives under such model has not being studied so far.\\
 The exchange of two assets can be used to  hedge against the changes in price of  underling assets by betting on the difference between both. The price of such instruments has been first considered in \cite{marg} under a bivariate Black-Scholes model, where a closed-form formula for the pricing is provided.  The results have been extended in \cite{chea, chea2} to  price  the exchange in the case of a  jump-diffusion model, while  \cite{exchir} have considered the pricing of the derivative under  stochastic interest rates.\\
    On the other hand  it is well known that constant correlation, constant volatilities  and continuous trajectories  are features not supported by empirical evidence. Some dynamic stochastic processes for the covariance have been previously proposed, see for example \cite{Fonseca} for the popular Wishart model, \cite{pigor} for an Ornstein-Uhlenbeck Levy type model and \cite{lipcsv} for  a simple  model based on a linear combination of Cox-Ingersol-Ross processes and finally an extension of \cite{barn2}  to a multivariate setting proposed in \cite{BN-nicolato}. We study its integrated characteristic function, moments and pricing  under the latter.\\
  As a closed-form pricing formula is not available when stochastic covariance and random jumps  are considered, approximations based on  polynomial expansions of the  price, after conditioning on the former, allow for  efficient and accurate calculations.\\
 Starting with a pioneer idea in \cite{hullwhitetaylor}, Taylor developments have been taken into account to compute the price of spread options and other multivariate contracts. For example, a second order Taylor  expansion has been successfully used in \cite{lidengzhou, lidengzhou2010} to price spread options under a multivariate Black-Scholes model. \\
  As a closed-formula  for exchange options is available in a Black-Scholes setting and it is possible, based on the knowledge of the  cumulated characteristic function, to compute mixed moments for the integrated covariance model, an approach following the same idea seems feasible to  be applied for the case studied in this paper. Moreover,  other   polynomial expansions such as a type I Chebyshev family can be considered, see  \cite{oliv},to obtain a uniform and more accurate approximation. See  \cite{grass} for an application of Chebyshev polynomials in other context of models and derivatives.\\
  Our approach is in essence, a combination of conditioning, polynomial expansions, FFT inversion and the existence of a closed-form for the price in a Black-Scholes setting that allows to value exchange options under a more realistic model with stochastic correlation in the underlying assets.\\
 The organization of the paper is the following:\\
In section 2 we introduce the main notations and discuss  the pricing of  the  exchange option by  polynomial expansions. In section 3 we define the Ornstein-Uhlenbeck covariance model following \cite{BN-nicolato}, compute the characteristic function of the integrated process and its moments. In section 4 we discuss algorithms and implementation of the method, while  numerical results allowing a comparison between the price obtained by Monte Carlo and polynomial approximations are shown in section 5. Proof of the theoretical results are deferred to the appendix.
 \section{Pricing exchange options in models with stochastic covariance }
Fist, we introduce some notations. We denote by $C_l$ a  matrix having  ones in position $(l,l)$ and zeros otherwise. For a matrix $A$ its trace is denoted by $tr(A)$ and its transpose by $A'$. For a vector $V$ the expression $diag(V)$ denotes a diagonal  matrix whose elements in the diagonal are the components of $V$. For two vectors $x$ and $y$, $xy$ represents its scalar product.\\
When $l$ is an integer number, $D^{l}$ represents the $l$-th order derivative operator. To simplify notations we make $D^1=D$. \\
Let  $(\Omega ,\mathcal{F},\mathcal{P}, (\mathcal{F}_{t})_{t \geq 0})$ be a filtered probability space. We denote  by $\mathcal{Q}$ an equivalent martingale  measure(EMM), and by  $r$ the (constant) interest rate or a vector with components equal to $r$.  The filtration $(\mathcal{F}_{t})_{t \geq 0}$ is assumed to verify the usual conditions,  i.e. they are right-continuous containing all events of probability zero.\\
 The $\sigma$-algebra  $\mathcal{F}^{X_t}$ is defined for any $t>0$ as  the  $\sigma$-algebra generated by the random variables $(X_s)_{ 0 \leq s \leq t}$.\\
 Also,  we define the increments of the process $(X_t)_{ t \geq 0}$ as $\Delta X_t= X_t-\lim_{s \uparrow t} X_s$. For two squared integrable semi-martingales $X$ and $Y$, $<X,Y>$ defines their quadratic covariation process. The functions $\varphi_X(u)$ and $\varphi_X(u, a, b)$ represent respectively the characteristic function of the random variable $X$ and the characteristic function of the random variable constrained to the interval $[a,b]$, both  under the chosen EMM.\\
 A  two-dimensional adapted stochastic process $(S_t)_{t \geq 0}=(S^{(1)}_t,S^{(2)}_t)_{t \geq 0}$, where their components represent prices of certain assets, is defined on the filtered probability space.\\
 We describe the processes of prices as follows:
 \begin{equation}\label{eq:prices1}
 S^{(j)}_t=S^{(j)}_0 exp(Y^{(j)}_t)\;\;j=1,2.\;\;
 \end{equation}
where $Y=(Y^{(1)}_t,Y^{(2)}_t)_{t \geq 0}$ is the process of log-prices.\\
We assume that the process of  log-prices  has a dynamic under  $\mathcal{Q}$  given by:
\begin{equation}\label{dynvector}
   dY_t=(r-q-\frac{1}{2} diag( \Sigma_t)) dt+ \Sigma_t^{\frac{1}{2}}  dB_t
\end{equation}
 while $(\Sigma _t^{1/2})_{0 \leq t \leq T}$ is a matrix-valued stochastic process  such that \\$\Sigma _t^{1/2} \Sigma _t^{1/2}=\Sigma _t$. Its components are $(\sigma_t)_{jk}$ for $0 \leq j,k \leq 2$.\\
 Under $\mathcal{Q}$, the process $(B_t)_{t \geq 0}=(B_t^{(1)},B_t^{(2)})_{t \geq 0}$ is a two-dimensional standard Brownian motion  with independent components.  The vector $q=(q_1,q_2)$ represents dividends on both assets.
The conditional joint distribution of $Y_T$ and its characteristic function are given in the  elementary lemma below.
\begin{lemma}\label{returndistb}
Let  $(S_t)_{0 \leq t \leq T}$ be a process driven by equations (\ref{eq:prices1}) and (\ref{dynvector}) under an EMM $\mathcal{Q}$. Then, conditionally on $\mathcal{F}^{\Sigma_T}$, the random variable $Y_T$ follows a bivariate normal distribution. More precisely:
 \begin{equation*}\label{eq:normaltdep}
Y_{T}\sim N \left((r-q)T-\frac{1}{2}diag(\Sigma^+_T),  \Sigma^+_T  \right)
\end{equation*}
where $\Sigma^+_T=(\sigma^+_T)_{jk}$ has components
\begin{equation*}
  (\sigma^+_T)_{jk}=\int_0^T (\sigma_t)_{jk}\;dt,\;\;\text{for} \;0 \leq j,k \leq 2
\end{equation*}
In particular, for a constant covariance process  $\Sigma^{+}_T=\Sigma T$.\\
Moreover, the  characteristic function of $Y_t$ is:
\begin{equation*}\label{eq:chfy}
  \varphi_{Y_T}(u)=exp(iu(r-q)) \varphi_{\Sigma^+_T}(-\frac{1}{2}\theta(u))
\end{equation*}
where:
\begin{equation*}
  \theta(u)= \left(\begin{array}{cc}
              u_1(1-i u_1) &  -i u_1 u_2 \\
              -i u_1 u_2 & u_2(1-i u_2)
            \end{array}
            \right)
\end{equation*}
\end{lemma}
\begin{proof}
From equation (\ref{dynvector}) we have:
\begin{eqnarray}\label{dyn-logretint}
 Y_T&=&(r-q)T-\frac{1}{2}  diag(\Sigma^+_T)+\int_0^T \Sigma_t^{1/2}  dB_t
\end{eqnarray}
The third term in the equation above follows a bivariate normal distribution, conditionally on $\mathcal{F}^{\Sigma_T}$, with zero mean and elements of the covariance matrix given by:
\begin{eqnarray*}
&& Var \left((\int_0^T (\Sigma_t^{1/2})  dB_t)_j/\mathcal{F}^{\Sigma_T} \right)= \left <\int_0^T (\Sigma_t^{1/2})_{j1}  dB^{(1)}_t+\int_0^T (\Sigma_t^{1/2})_{j2}  dB^{(2)}_t \right>\\
&=& \int_0^T [(\Sigma_t^{1/2})^2_{j1}+(\Sigma_t^{1/2})^2_{j2}]dt=\sigma^{jj+}_T,\; j=1,2.
\end{eqnarray*}
Similarly:
\begin{eqnarray*}
 cov \left (\left(\int_0^T \Sigma_t^{1/2} dB_t \right)_1,\left(\int_0^T \Sigma_t^{1/2} dB_t \right)_2/ \mathcal{F}^{\Sigma_T}\right)   &=& \sigma^{12+}_T
\end{eqnarray*}
On the other hand, from equation (\ref{dyn-logretint}) and the conditional normality of the log-prices:
     \begin{eqnarray*}
       \varphi_{Y_T}(u) &=& E_{\mathcal{Q}}\left[ E_{\mathcal{Q}} \left( exp(iu Y_T)/ \mathcal{F}^{\Sigma_T} \right)\right]\\
       &=& exp(iu(r-q)) E_{\mathcal{Q}}  \left[ exp(-\frac{1}{2}i u\; diag(\Sigma^{+}_T)-\frac{1}{2}u   \Sigma^{+}_T u')\right]\\
               &=& exp(iu(r-q)) \varphi_{\Sigma^+_T}(-\frac{1}{2}\theta(u))
     \end{eqnarray*}
\end{proof}
 The payoff of a European  exchange option, with maturity at time $T>0$ is
\begin{equation*}\label{eq:payoffspread}
      h(Y_T)=(cS^{(1)}_0 exp(Y^{(1)}_t)-m S^{(2)}_0 exp(Y^{(2)}_t))_+
\end{equation*}
where $m$ is the number of assets of type two exchanged against $c$ assets of type one.\\
A closed-form  formula for the  price of an exchange  under a bivariate Black-Scholes model, i.e. the model given by equation (\ref{dynvector}) with a constant covariance, starting at $t=0$ has been found in \cite{marg}. This price, called \textit{Margrabe price}, is denoted by $C_M:=C_M(\Sigma)$.\\
On the other hand, the price of the exchange option under the full model, i.e. the one driven  by equation equation (\ref{dynvector}), after conditioning on $\mathcal{F}^{\Sigma_T}$ is denoted by $C_{MT}:=C_{MT}(\Sigma_T^+)$.\\
Notice that when conditioning on the covariance process the price process becomes a bivariate Gaussian model with deterministic and time-dependent volatilities and correlation.\\
Both prices are related by $C_{MT}(\Sigma_T^+)=C_M(\frac{1}{T}\Sigma_T^+)$, as the price of the later is equivalent to the Margrabe price with constant  covariance matrix $\frac{1}{T}\Sigma^+_T$.  \\
Additionally, we denote by $C_{MS}$ the unconditional price of the contract.  To be more precise, the price of the exchange under the model (\ref{dynvector})  is  given by:
\begin{equation} \label{eq:pricepi}
 C_{MS}= E_{\mathcal{Q}} \left[ C_{MT} \left( \Sigma _{T}^{+}\right) \right]
\end{equation}
where:
\begin{equation}\label{eq:condprice}
C_{MT} \left( \Sigma _{T}^{+}\right)= exp^{-rT} E_{\mathcal{Q} }\left[ (c S_0^{(1)}exp(-(r-q_1)T)exp(Y_T^{(1)})-mS_0^{(2)}exp(-(r-q_2)T) exp(Y_T^{(2)}))_+|\mathcal{F}^{\Sigma_T }\right]
\end{equation}
is the price of the exchange contract after conditioning on $\mathcal{F}^{\Sigma_T}$.\\
From the   remark above and lemma \ref{returndistb} a simple extension of  Margrabe formula to the case of time-dependent deterministic covariance is given by:
\begin{equation*}\label{eq:pricemag}
    C_{MT}(\Sigma^+_T)=ce^{-(r-q_1)T}S^{(1)}_0 N(d_1)-me^{-(r-q_2)T}S^{(2)}_0 N(d_2)
\end{equation*}
\begin{eqnarray*}
    d_1&=&\frac{\log \left(\frac{c S^{(1)}_0}{m S^{(2)}_0} \right)+(q_1-q_2)T+\frac{1}{2}v^+_T }{\sqrt{v^+_T} }\\
    d_2&=&\frac{\log \left(\frac{c S^{(1)}_0}{m S^{(2)}_0} \right)+(q_1-q_2)T-\frac{1}{2}v^+_T }{\sqrt{v^+_T} }=d_1-\sqrt{v^+_T}
\end{eqnarray*}
with
$v^+_T=\sigma_T^{11+}+\sigma_T^{22+}-2 \sigma^{12+}_T$.
\begin{remark}
The conditional Margrabe price $C_{MT}$ depends on $\Sigma^+_T$ through the quantity $v^+_T$. Consequently we write $ C_{MT}(\Sigma^+_T)= C_{MT}(v^+_T)$.
\end{remark}
\subsection{Pricing by polynomial expansions}
In the general case there is not analogous to Margrabe pricing formula.  It is possible to approximate the price of the exchange by a suitable expansion of $C_{MT} \left( v^+_T \right)$  in  terms of Taylor polynomials around a point $v^*$, typically  around the mean value of the integrated process given by $v^*= E_{\mathcal{Q}} (v^+_T)$, or using a family of polynomials such as first type Chebyshev functions or cubic splines.  \newline
i)Taylor approximation.\\
  The one dimensional Taylor expansion of $n$-th order,  denoted $C_{MS}(v, v^*)$, around the value $v*$ is given by:
\begin{equation*}\label{eq:taylorvstar}
  C_{MS}(v,v^*)=\sum_{l=0}^n \frac{D^{l}C_{MT}(v^*)}{l!}(v-v^*)^l
\end{equation*}
A Taylor approximation of the price, taking into account equation (\ref{eq:pricepi}), is defined by:
\begin{equation}\label{eq:taylorapproxuni}
  \hat{C}_{MS}^{(n)}(v^*)=\sum_{l=0}^n \frac{D^{l}C_{MT}(v^*)}{l!}E_{\mathcal{Q}}(v^+_T-v^*)^l
\end{equation}
\begin{remark}
Notice that, in order to implement the approximation above we need the derivatives of the $C_{MT}(v, v^*)$  up to order $n$ and the mixed moments of the components in the integrated covariance matrix $\Sigma^+_T$. \\
\end{remark}
\begin{remark}
Sensitivities with respect to the parameters in the contract can be obtained in a similar way. For example, approximations of the \textit{deltas} are:
\begin{equation*}
  \frac{\partial \hat{C}_{MS}^{(n)}}{\partial s^{(j)}}(v^*)=\sum_{l=0}^n \frac{D^{l}}{l!}  \frac{\partial C_{MT}(v^*)}{\partial s^{(j)}}  E_{\mathcal{Q}}(v^+_T-v^*)^l,\;\;j=1,2.
\end{equation*}
\end{remark}
  ii) Approximation by cubic splines.\\
 On an interval $[a,b]$  we consider a partition $a=v_0 \leq v_1 \leq \ldots \leq  v_N  \leq b$.\\
   An approximation of $C_{MT}$ based on cubic splines is thus given by:
 \begin{eqnarray}\label{eq:splpoly}
  C^{spl}(v) &=& \sum_{j=0}^{N-1} \sum_{l=0}^3 \alpha_{l,j}  1_{[v_j,v_{j+1})} v^l
  \end{eqnarray}
  The coefficients $\alpha_{l,j}$ depend on the partition.  \\
  Additional conditions on the derivatives to smooth these curves are usually imposed. Namely, $D_{-}C_{MT}(v_j)= D_{+}C_{MT}(v_j), j=1,2,\ldots,N$, $l=1,2$, where $D_{-}C_{MT}(v_j)$ and $D_{+}C_{MT}(v_j)$ are respectively the derivatives from the left and the right of the function $C_{MT}$ at point $v=v_j$.\\
   Moreover, for end points in the interval we set $D^2 C_{MT}(a)=D^2 C_{MT}(y_N)=0$. See \cite{arc} for a general account on splines and its implementation.\\
    On the other hand this approach requires  the constrained moments of $v^+_T=tr(M \Sigma^+_T)$ up to order $n$ . To this end we first compute  the corresponding characteristic function of the covariance process constrained to $[a,b]$ by using Fourier inversion formula. Notice that:
 \begin{eqnarray}\nonumber
  \varphi_{v^+_T}(u)&=& \varphi_{\Sigma^+_T}(Mu)\\ \label{eq:chfconstnt}
   \varphi_{v^+_T}(u,a,b)&=& \varphi_{\Sigma^+_T}(Mu,a,b)=E_{\mathcal{Q}} \left[exp(i tr(M u \Sigma^+_T))1_{[a,b]}(tr(M  \Sigma^+_T))\right]
 \end{eqnarray}
 where the matrix $M$ is:
 \begin{equation*}
   M=\left( \begin{array}{cc}
                    1 & -1 \\
                    -1 & 1
                  \end{array}
   \right)
 \end{equation*}
   The constrained moments of $v^+_T$ can be obtained by differentiating equation (\ref{eq:chfconstnt}) with respect to $u$ and evaluating at $u=0$.\\
 Now, we replace the function $C_{MT}(v)$ in equation (\ref{eq:pricepi}) by its approximation given in equation (\ref{eq:splpoly}) to obtain the following estimated price:
   \begin{equation}\label{eq:pricespline}
   C^{spl}= \sum_{j=0}^{N-1} \sum_{l=0}^3 \alpha_{l,j} \tilde{m}_{v^+_T}(l,v_j,v_{j+1})
 \end{equation}
 where:
 \begin{equation}\label{}
   \tilde{m}_{v^+_T}(l,a,b)=E_{\mathcal{Q}}[(v^+_T-a)^l 1_{[a,b)}(v^+_T)],\; a,b \in \mathbb{R}
 \end{equation}
 are the constrained moments on $[a,b)$ of $v^+_T$ centered at $a$. Their calculation is discussed in section 4.
 \section{An Ornstein-Uhlenbeck stochastic covariance model}\label{covmod}
   Our  model is based on the general Ornstein-Uhlenbeck process with a  Levy Background  Noise Process (BDLP)  as studied in \cite{barn2}. It has been  extended to a multidimensional setting in \cite{BN-nicolato}. \\
    We define a matrix-valued correlation  process, based on  independent   Levy processes $F_t=(F_{t}^{(1)},F_{t}^{(2)})$ and $V_t=(V_{t}^{(1)},V_{t}^{(2)})$,  with respective characteristic exponents $\psi_F$ and $\psi_V$.\\
The covariance process is defined for any $t \geq 0$  by:
\begin{equation}\label{eq:covfactorpc}
     \Sigma_t= diag( F_t)+A\; diag(V_t)\; A'
\end{equation}
where $A=(a_{ij})$ is a $2 \times 2$ deterministic orthonormal loading  matrix.\\
Note that $F$ and $V$ correspond with  idiosyncratic and common factors respectively. Furthermore, we assume  $F^{(l)}$ and $V^{(l)}$ are   Ornstein-Ulenbeck Levy processes given by:
\begin{equation}\label{eq:ouf}
    dF_t^{(l)}=-\lambda_{F,l}F_t^{(l)}dt+ dZ_{\lambda_{F,l}t}^{(F,l)}
\end{equation}
\begin{equation}\label{eq:ouv}
    dV_t^{(l)}=-\lambda_{V,l}V_t^{(l)}dt+ dZ_{\lambda_{V,l}t}^{(V,l)}
\end{equation}
 with BDLP given respectively by $(Z_{\lambda_{F,l}t}^{(F,l)})$ and $(Z_{\lambda_{V,l}t}^{(V,l)})$, $\lambda_{F,l}>0, \lambda_{V,l}>0, l=1,2$.\\
  After applying Ito formula we have that the integrated processes corresponding to equations (\ref{eq:ouf}) and (\ref{eq:ouv}) are given by:

 \begin{equation}\label{eq:intouf}
    F_t^{(l,+)}=\lambda_{F,l}^{-1}(1-exp(-\lambda_{F,l}t)) F_0^{(l)}+\lambda_{F,l}^{-1} \int_0^t (1-exp(-\lambda_{F,l}(t-s)))dZ^{F,l}_{\lambda_{F,l}s}
\end{equation}
\begin{equation}\label{eq:intouv}
    V_t^{(l,+)}=\lambda_{V,l}^{-1}(1-exp(-\lambda_{V,l}t))V_0^{(l)}+\lambda_{V,l}^{-1} \int_0^t (1-exp(-\lambda_{V,l}(t-s)))dZ^{V,l}_{\lambda_{V,l}s}
\end{equation}
We consider  inverse Gaussian subordinators with respective characteristic exponents :
\begin{equation}\label{eq:charexpf}
\Psi_{Z_t^{F,l}}(\theta)=-a_{F,l} \left(\sqrt{-2i \theta+b^2_{F,l}} -b_{F,l}\right)
\end{equation}
\begin{equation}\label{eq:charexpv}
\Psi_{Z_t^{V,l}}(\theta)=-a_{V,l} \left(\sqrt{-2i \theta+b^2_{V,l}} -b_{V,l}\right)
\end{equation}
The integrated covariance process is given by:
 \begin{equation}\label{eq:intcovfactorpc}
\Sigma _{t}^{+}=\int_{0}^{t} (diag(F_s)+A\; diag(V_s)\; A')\;  ds=diag(F^+_t) +A\; diag(V^+_t) A'
\end{equation}
Its characteristic function is computed in the proposition below:
\begin{theorem}\label{chfbnmodel}
Let  $\Sigma _{t}^{+}$ be the integrated covariance processes  defined by equation (\ref{eq:intcovfactorpc}), with $F=(F_t)_{t \geq 0}$ and $V=(V_t)_{t \geq 0}$ following Ornstein-Ulenbeck processes having initial deterministic values $F_0$ and $V_0$ and independent Inverse Gaussian subordinators as BDLPs.\\
 Denote by  $\varphi_{\Sigma_t^+}$  its  characteristic functions, let $\theta=(\theta_{kj})_{k,j=1,2}$ be a $2 \times 2$ matrix and $\tilde{\theta} \in \mathbb{R}$. Then, for $\theta \neq 0$:
\begin{eqnarray}\label{eq:intcharfuncorre}\notag
    \varphi_{\Sigma_t^+}(\theta)&=&exp \left( K^+_1(\theta)+K^+_2(\theta) \right)
\end{eqnarray}
with:
\begin{eqnarray} \nonumber
    K^+_1(\theta)&=&i \sum_{l=1}^2 \theta_{ll}\lambda_{F,l}^{-1}(1-exp(-\lambda_{F,l}t))F_0^{(l)}+\sum_{l=1}^2  I_F^{(l)}(\lambda_{F,l}t, \theta_{ll})\\ \nonumber
   K^+_2(\theta)&= &i\sum_{l=1}^2 tr(  \theta AC_{l}A') \lambda_{V,l}^{-1}((1-exp(-\lambda_{V,l}t))V_0^{(l)}+ \sum_{l=1}^2   I_V^{(l)}(\lambda_{V,l}t,tr(  \theta AC_{l}A'))\\ \nonumber
  && \\     \nonumber
    I_F^{(l)}(\lambda_{F,l}t,\tilde{\theta})&=&{-\frac {2\,a_{F,l}}{\sqrt {i \lambda_{F,l}}}}\left[ -T^{F,l}_2(\tilde{\theta})+\sqrt{i \lambda_{F,l}}b_{F,l}+ \frac{i}{2}T^{F,l}_1(\tilde{\theta}) G^{F,l}(\tilde{\theta}) \right]+\lambda_{F,l} a_{F,l} b_{F,l} t\\ \label{eq:igintf}
    &&\\ \nonumber
       I_V^{(l)}(\lambda_{V,l}t,\tilde{\theta})&=&{-\frac {2\,a_{V,l}}{\sqrt {i \lambda_{V,l}}}}\left[ -T^{V,l}_2(\tilde{\theta})+\sqrt{i \lambda_{V,l}}b_{V,l}+\frac{i}{2}T^{V,l}_1(\tilde{\theta}) G^{V,l}(\tilde{\theta})\right]+\lambda_{V,l} a_{V,l} b_{V,l} t \\ \label{eq:igintv}
       &&
\end{eqnarray}
 and
\begin{equation*}\label{eq:g}
  G^{F,l}(\tilde{\theta})=\log \left( exp(-\lambda_{F,l}t) \frac{(T^{F,l}_1(\tilde{\theta})+i \sqrt{i\lambda_{F,l}}\,b_{F,l})^2 }{(T^{F,l}_1(\tilde{\theta})+i T^{F,l}_2(t,\tilde{\theta}))^2}\right)
 \end{equation*}
 \begin{eqnarray*}
    T^{F,l}_1(\tilde{\theta})&=& \sqrt {-2\,\tilde{\theta}-i\lambda_{F,l}\,{b^{2}_{F,l}}}\\
T^{F,l}_2(t,\tilde{\theta})&=&\sqrt {2\,\tilde{\theta}(1-{{\rm e}^{-\lambda_{F,l}\,t}})+i\lambda_{F,l}\,{b^{2}_{F,l}}} \\
 \end{eqnarray*}
Analogous expressions for $T^{V,l}_1, T^{V,l}_2$ and $G^{V,l}$ are defined after replacing $F$ by $V$.
\end{theorem}
\begin{proof}
See appendix.
 \end{proof}
Moments of the integrated process can be obtained from the derivatives of the integrated characteristic function evaluated at zero. To this end we need to compute the derivatives of  expressions (\ref{eq:igintf}) and (\ref{eq:igintv}).\\ For simplicity we provisionally drop the dependence on $V$ and $F$. Notice that $I(\lambda t,\tilde{\theta}\lambda^{-1}(1-exp(-\lambda t+s)))$ is differentiable with respect to $\tilde{\theta}$ in a vicinity of zero. Moreover, at  points $\tilde{\theta}$ different from zero:
\begin{eqnarray*}
  \frac{\partial \Psi_Z(\tilde{\theta}\lambda^{-1}(1-exp(-\lambda t+s)))}{\partial \tilde{\theta}} &=& -\frac{a(1-exp(-\lambda t+s))}{\sqrt{i \lambda}\sqrt{2 \tilde{\theta}(1-exp(-\lambda t+s))+i \lambda b^2 }}
\end{eqnarray*}
For the case $\tilde{\theta}=0$ we take into account that  $\Psi_Z(0)=0$ to have:
\begin{eqnarray*}
    \frac{\partial \Psi_Z(\tilde{\theta}\lambda^{-1}(1-exp(-\lambda t+s)))}{\partial \tilde{\theta}}|_{\tilde{\theta}=0} &=&\frac{a}{\sqrt{\lambda}} \lim_{\tilde{\theta} \rightarrow 0}\frac{i (1-exp(-\lambda t+s)) }{\sqrt{-2 i \tilde{\theta}(1-exp(-\lambda t+s))+\lambda b^2}}\\
   &=&\frac{i a(1-exp(-\lambda t+s))}{\lambda b}
\end{eqnarray*}
The fact that the function $\Psi_Z$ is continuously differentiable on a vicinity of zero and continuous on the variable  $s$ on the interval $[0,\lambda t]$ allows to interchange derivative and integration by Lebesgue Dominated Convergence Theorem. Therefore, for $\tilde{\theta} \neq 0$:
\begin{eqnarray}\label{eq:derint1}
 \frac{\partial I(\lambda t, \tilde{\theta})}{\partial \tilde{\theta}} &=& -\frac{a}{\sqrt{i \lambda}} \int_0^{\lambda t} \frac{1-exp(-\lambda t+s)}{\sqrt{2 \tilde{\theta}(1-exp(-\lambda t+s))+i \lambda b^2 }}\;ds
\end{eqnarray}
At $\tilde{\theta}=0$:
\begin{eqnarray*}
 \frac{\partial I(\lambda t, \tilde{\theta})}{\partial \tilde{\theta}}|_{\tilde{\theta}=0} &=& i \frac{a}{ \lambda b}(\lambda t-(1-exp(-\lambda t)))
\end{eqnarray*}
Moreover, for $n \geq 2$ the $n$-th derivative is obtained as:
\begin{eqnarray*}
  \frac{\partial^n I(\lambda t,\tilde{\theta})}{\partial \tilde{\theta}^n}&=& (-1)^n \prod_{k=2}^n (2k-3)\frac{a}{\sqrt{i \lambda} } \int_0^{\lambda t} \frac{(1-exp(-\lambda t+s))^n}{(2 \tilde{\theta}(1-exp(-\lambda t+s))+i \lambda b^2)^{\frac{2n-1}{2}}}\;ds
  \end{eqnarray*}
  and evaluating at $\tilde{\theta}=0$:
  \begin{eqnarray*}
  \frac{\partial^n I(\lambda t,\tilde{\theta})}{\partial \tilde{\theta}^n}|_{\tilde{\theta}=0} &=& (-1)^n \prod_{k=2}^n (2k-3)\frac{a}{(i\lambda)^n b^{2n-1}}\int_0^{\lambda t} (1-exp(-\lambda t+s))^n\;ds\\
    &=&  (-1)^{n}\prod_{k=2}^n (2k-3)\frac{a}{(i\lambda)^n b^{2n-1}} \left[\sum_{k=1}^n \left( \begin{array}{c}
                                                                                                                    n \\
                                                                                                                    k
                                                                                                                  \end{array}
   \right)(-1)^{k}\frac{(1-exp(-k\lambda t))}{k}+\lambda t\right]
   \end{eqnarray*}
\begin{proposition}
Let  $(\Sigma _{t}^{+})_{t \geq 0}$ be the integrated covariance processes  given by equation (\ref{eq:intcovfactorpc}), where $F=(F_t)_{t \geq 0}$ and $V=(V_t)_{t \geq 0}$ follow Ornstein-Ulenbeck processes with initial deterministic values $F_0$ and $V_0$ and independent Inverse Gaussian subordinators as BDLPs. Then, the first two  moments of the elements in $(\Sigma^+_t)_{t \geq 0}$  are  given by:
 \begin{eqnarray} \nonumber
     E_{\mathcal{Q}}(\sigma^{kk+}_t) &=& \lambda_{F,k}^{-1}(1-exp(-\lambda_{F,k}t))F_0^{(k)}-i \frac{\partial  I_F^{(k)}(\lambda_{F,k} t, \theta_{kk})}{\partial \theta_{kk}}|_{\theta_{kk}=0}\\ \nonumber
    &+&  \sum_{l=1}^2 a^2_{kl} \lambda_{V,l}^{-1}(1-exp(-\lambda_{V,l}t))V_0^{(l)}-i  \sum_{l=1}^2 a^2_{kl} \frac{\partial I_V^{(l)}(\lambda_{V,l}t,tr(\theta AC_{l}A'))}{\partial \theta_{kk}}|_{\theta=0}\\
      \label{eq:moment1gralkkf}
   \end{eqnarray}
   \begin{eqnarray} \nonumber
     E_{\mathcal{Q}}(\sigma^{12+}_t) &=&   \sum_{l=1}^2 a_{1l}a_{2l} \lambda_{V,l}^{-1}(1-exp(-\lambda_{V,l}t))V_0^{(l)}-i  \sum_{l=1}^2  a_{1l}a_{2l} \frac{\partial I_V^{(l)}(\lambda_{V,l}t,tr(\theta AC_{l}A'))}{\partial \theta_{12}}|_{\theta=0}\\
      \label{eq:moment1gral12f}
   \end{eqnarray}
   where for $k,l,j=1,2$:
   \begin{eqnarray*}\label{eq:der1int}
\frac{\partial I_F^{(k)}( \lambda_{F,k} t,\tilde{\theta})}{\partial \tilde{\theta}}(\lambda_{F,k} t,0)&:=& \frac{\partial I_F^{(k)}( \lambda_{F,k} t,\tilde{\theta})}{\partial \tilde{\theta}}|_{\tilde{\theta}=0}\\
 &=& \frac{i a_{F,k}}{ \lambda_{F,k} b_{F,k}}(\lambda_{F,k} t-(1-exp(-\lambda_{F,k} t))\\ \label{eq:der1int}
\frac{\partial I_V^{(l)}( \lambda_{V,l} t,tr(\theta AC_{l}A'))}{\partial \theta_{kj}}(\lambda_{V,l} t, 0) &:=&\frac{\partial I_V^{(l)}( \lambda_{V,l} t,tr(\theta AC_{l}A'))}{\partial \theta_{kj}}|_{\theta=0}\\
&=& \frac{i a_{V,l}}{ \lambda_{V,l} b_{V,l}}(\lambda_{V,l} t+exp(-\lambda_{V,l} t)-1)
\end{eqnarray*}
Moreover, for $k,l=1,2$:
  \begin{eqnarray} \nonumber
     E_{\mathcal{Q}}(\sigma^{kk+}_t)^2 &=& -\left(i \lambda_{F,k}^{-1}(1-exp(-\lambda_{F,k}t))F_0^{(k)} + \frac{\partial I^{(k)}_F( \lambda_{F,k} t,0)}{\partial \theta_{kk}} \right.\\ \nonumber
          &+& i \sum_{l=1}^2 a^2_{kl} \lambda_{V,l}^{-1}(1-exp(-\lambda_{V,l}t))V_0^{(l)}+\left. \sum_{l=1}^2 a^2_{kl} \frac{\partial I_V^{(l)}(\lambda_{V,l}t,0)}{\partial \theta_{kk}}   \right)^2\\ \nonumber
   &-&  \left(\frac{\partial^2 I^{(l)}_F( \lambda_{F,k} t,\theta_{kk})}{\partial \theta^2_{kk}}|_{\theta_{kk}=0} +\sum_{l=1}^2 a^4_{kl} \frac{\partial^2 I_V^{(l)}(\lambda_{V,l}t,0)}{\partial \theta^2_{kk}} \right)\\
      && \label{eq:moment2gralkkf2}
   \end{eqnarray}
\begin{eqnarray}\notag
E_{\mathcal{Q}}[(\sigma^{12+}_t)^2]  &=& - \sum_{l=1}^2 a^2_{1l}a^2_{2l}\frac{\partial^2 I^{(l)}_V(\lambda_{V,l} t,0)}{\partial \theta^2_{12}} \\ \notag
 &-&  \left(i \sum_{l=1}^2 a_{1l}a_{2l} \lambda_{V,l}^{-1}(1-exp(-\lambda_{V,l}t))V_0^{(l)}+\sum_{l=1}^2 a_{1l}a_{2l} \frac{\partial I^{(l)}_V(\lambda_{V,l} t,0)}{\partial \theta_{12}}| \right)^2 \\   \label{eq:moment212uneq}
    &&
    \end{eqnarray}
   \begin{eqnarray}\nonumber
E_{\mathcal{Q}}(\sigma^{kk+}_t \sigma^{12+}_t)&=&   -\sum_{l=1}^2 a_{kl}^2a_{1l}a_{2l} \frac{\partial^2 I^{(l)}_V(\lambda_{V,l} t,0)}{\partial \theta^2_{kk}} \\ \notag
&-&\left( i   \lambda_{F,k}^{-1}(1-exp(-\lambda_{F,k}t))F_0^{(k)}+\frac{\partial I_F^{(k)}}{\partial \theta_{kk}}( \lambda_{F,k} t,0)| \right.\\ \nonumber
     &+& \left. i \sum_{l=1}^2 a_{kl}^2 \lambda_{V,l}^{-1}(1-exp(-\lambda_{V,l}t))V_0^{(l)}+\sum_{l=1}^2 a_{kl}^2 \frac{\partial I^{(l)}_V(\lambda_{V,l} t,tr(\theta AC_{l}A'))}{\partial \theta_{kk}}|_{\theta=0}\right)\\ \nonumber
&&\left( i \sum_{l=1}^2 a_{1l}a_{2l} \lambda_{V,l}^{-1}(1-exp(-\lambda_{V,l}t))V_0^{(l)}+  \sum_{l=1}^2 a_{1l} a_{2l}  \frac{\partial I^{(l)}_V(\lambda_{V,l} t,0)}{\partial \theta_{12}}|\right) \label{eq:moment2kk12uneq2}
 \end{eqnarray}
for $k=1,2$.
\begin{eqnarray}\notag
E_{\mathcal{Q}}(\sigma^{11+}_t \sigma^{22+}_t) &=& -\sum_{l=1}^2 a^2_{1l}a^2_{2l}\frac{\partial^2 I_V^{(l)}(\lambda_{V,l}t,0}{\partial \theta_{11}\partial \theta_{22}} \\ \nonumber
&-& \left( i \lambda_{F,1}^{-1}(1-exp(-\lambda_{F,1}t))F_0^{(1)}+\frac{\partial I_F^{(1)}( \lambda_{F,1} t,0)}{\partial \theta_{11}}| \right.\\ \nonumber
&+& i \left. \sum_{l=1}^2 a_{1l}^2 \lambda_{V,l}^{-1}(1-exp(-\lambda_{V,l}t))V_0^{(l)}+\sum_{l=1}^2 a_{1l}^2 \frac{\partial I^{(l)}_V(\lambda t,0)}{\partial \theta_{11}}\right)\\ \nonumber
&&  \left( i \lambda_{F,2}^{-1}(1-exp(-\lambda_{F,2}t))F_0^{(2)}+\frac{\partial I_F^{(2)}( \lambda_{F,2} t,0)}{\partial \theta_{22}}\right.\\ \nonumber
&+& i \left. \sum_{l=1}^2 a_{2l}^2 \lambda_{V,l}^{-1}(1-exp(-\lambda_{V,l}t))V_0^{(l)}+\sum_{l=1}^2 a_{2l}^2 \frac{\partial I^{(l)}_V(\lambda_{V,l} t,0)}{\partial \theta_{22}}\right)\\ \label{eq:moment1122uneq3}
   \end{eqnarray}
where for $k,l,j=1,2$:
\begin{eqnarray*}\notag
\frac{\partial^2 I_F^{(k)}(\lambda_{F,k} t,\theta_{kk})}{\partial \theta_{kk}^2}(\lambda_{F,k} t, 0)&:=& \frac{\partial^2 I_F^{(k)}(\lambda_{F,k} t,\theta_{kk})}{\partial \theta_{kk}^2}|_{\theta_{kk}=0}\\
&=&\frac{a_{F,k}}{\lambda_{F,k}^2 b_{F,k}^3}\left[ \lambda_{F,k} t-2(1-exp(-\lambda_{F,k} t))+\frac{1}{2}(1-exp(-2\lambda_{F,k} t))\right]
   \end{eqnarray*}
\begin{eqnarray*}\notag
\frac{\partial^2 I^{(l)}_V(\lambda t,tr(\theta AC_{l}A'))}{\partial \theta^2_{kj}}|_{\theta=0}&=&\frac{a_{V,l}}{\lambda_{V,l}^2 b_{V,l}^3}\left[ \lambda_{V,l} t-2(1-exp(-\lambda_{V,l} t))+\frac{1}{2}(1-exp(-2\lambda_{V,l} t))\right]
   \end{eqnarray*}
\end{proposition}
\begin{proof}
See appendix.
\end{proof}
The constrained moments  of $v^+_T$ are needed in the cubic spline approaches. They are obtained via the constrained characteristic function in the proposition above.
\begin{proposition}\label{prop:constchf3}
Let  $\Sigma _{t}^{+}$ be the integrated covariance processes  defined by equation (\ref{eq:intcovfactorpc}), with $F=(F_t)_{t \geq 0}$ and $V=(V_t)_{t \geq 0}$ following Ornstein-Ulenbeck processes having initial deterministic values $F_0$ and $V_0$ and independent Inverse Gaussian subordinators as BDLPs.\\
 Denote by  $\varphi_{v^+_T}(u,a,b)$  the constrained  characteristic function of $v^+_T= tr(M \Sigma^+_T)$. Then:
\begin{eqnarray}\label{eq:constchf2}
    \varphi_{v^+_T}(u,a,b)&=& - \frac{i}{2 \pi} \int_{\mathbb{R}} f(y,a,b) g(u-y)\;dy
\end{eqnarray}
where $g(x)=exp(K^+_1(Mx)+K^+_2(Mx))$ and
\begin{equation*}
  f(y,a,b)=\left \{ \begin{array}{cc}
             \frac{exp(iby)-exp(iay)}{y}    & y \neq 0 \\
                i(b-a) & y=0
              \end{array}
   \right.
\end{equation*}
Moreover, derivatives of the constrained characteristic function with respect to $u$ evaluated at zero can be computed as:
\begin{eqnarray}\nonumber
 && D^n \varphi_{v^+_T}(u,a,b)|_{u=0}= - \frac{i}{2 \pi} \sum_{j=0}^{n-1} \left(\begin{array}{c}
                         n-1\\
                          j
                        \end{array} \right)\\ \nonumber
 &&\int_{\mathbb{R}} f(y,a,b)[D^{j+1}K^+_1(-My)+ D^{j+1} K^+_2(-My)]D^{n-j-1}g(-y)\;dy \\ \label{eq:dercnschf}
                        &&
\end{eqnarray}
for $j=2,3,\ldots, n-1$.\\
Here $D^n$ is the $n$-th derivative with respect to the variable $u$ and $D^j K^+_l(M(u-y))$ is the $j$-th derivative of $K^+_l(M(u-y))$, also with respect to $u$ and evaluated at $u=0$.
\end{proposition}
\begin{proof}
See appendix.
\end{proof}
 \section{Implementing polynomial  expansions}
In this section we precise the pricing formulas under the two approximations considered .\\
 First, we implement the Taylor method based on equation (\ref{eq:taylorapproxuni}). To this end we first compute the Margrabe price $C_{MT}(v)$ under a model with time-dependent and deterministic volatilities and correlation, together with   its derivatives evaluated at $v=v^*$.\\
 In order to simplify notations we write:
\begin{eqnarray*}
  M_1 &=& c \;exp(-(r-q_1)T)S_0^{(1)},   M_2 = m\; exp(-(r-q_2)T)S_0^{(2)}\\
  M_3 &=& \log \left(\frac{c S^{(1)}_0}{m S^{(2)}_0} \right)+(q_1-q_2)T \\
\end{eqnarray*}
Then, by elementary calculations it follows that:
\begin{eqnarray*}
 D^{k}d_1(v^*)&=&M_3 T^{-\frac{1}{2}}(-1)^k \prod_{j=0}^{k-1} (\frac{1}{2}+j)(v^*)^{-\frac{1}{2}-k}+\frac{1}{2} \sqrt{T} \prod_{j=0}^{k-1} (\frac{1}{2}-j)(v^*)^{\frac{1}{2}-k}
   \end{eqnarray*}
Hence, differentiating the Margrabe formula:
\begin{eqnarray}\nonumber
D^{k}C_{MT}(v^*)&=& M_1 \sum_{j=0}^{k-1} \left( \begin{array}{c}
                                                     k-1 \\
                                                     j
                                                   \end{array}
  \right)D^{j}f_Z(d_1(v^*))D^{k-j}d_1(v^*)\\ \nonumber
  &-& M_2 \sum_{j=0}^{k-1} \left( \begin{array}{c}
                                                     k-1 \\
                                                     j
                                                   \end{array}
  \right)D^{j}[ f_Z(d_1(v^*)-\sqrt{v^*}\sqrt{T})]D^{k-j}d_1(v^*)\\ \nonumber
  &+& M_2 \sqrt{T} \sum_{j=0}^{k-1} \left( \begin{array}{c}
                                                     k-1 \\
                                                     j
                                                   \end{array}
  \right)D^{j}[ f_Z(d_1(v^*)-\sqrt{v^*}\sqrt{T})]\prod_{l=0}^{k-j} (\frac{1}{2}-l)(v^*)^{\frac{1}{2}-k+j}\\ \label{eq:dercmt}
  &&
\end{eqnarray}
Therefore, the price based on the first order Taylor expansion can be computed as:
\begin{eqnarray*}
  \hat{C}^{(1)}_{MS}(v^*) &=& A_0^{(1)}C_M(v^*)+D^{1}C_M(v^*)E_{\mathcal{Q}}(\sigma^{11+}_T)+D^{1}C_M(v^*)E_{\mathcal{Q}}(\sigma^{22+}_T)\\
  &-&2  D^{1}C_M(v^*)E_{\mathcal{Q}}(\sigma^{12+}_T)
\end{eqnarray*}
where $A_0^{(1)}=1-v^*D^{1}$.\\
For the second order expansion we  compute:
\begin{eqnarray}\nonumber
  E_{\mathcal{Q}}(v^+_T-v^*)^2 &=& E_{\mathcal{Q}}(\sigma^{11+}_T)^2+2 E_{\mathcal{Q}}(\sigma^{11+}_T \sigma^{22+}_T)-4  E_{\mathcal{Q}}(\sigma^{11+}_T \sigma^{12+}_T)\\ \nonumber
  &+&E_{\mathcal{Q}}[(\sigma^{22+}_T)^2]-4  E_{\mathcal{Q}}(\sigma^{22+}_T \sigma^{12+}_T)+4  E_{\mathcal{Q}}[(\sigma^{12+}_T)^2]\\ \nonumber
  &-&2v^*E_{\mathcal{Q}}(\sigma^{11+}_T)-2v^*E_{\mathcal{Q}}(\sigma^{22+}_T)+4  v^* E_{\mathcal{Q}}(\sigma^{12+}_T)+(v^*)^2\\ \label{eq:mom2}
  &&
\end{eqnarray}
\begin{figure}[hb!]
\begin{center}
\subfigure[]{
\resizebox*{3.5cm}{!}{\includegraphics{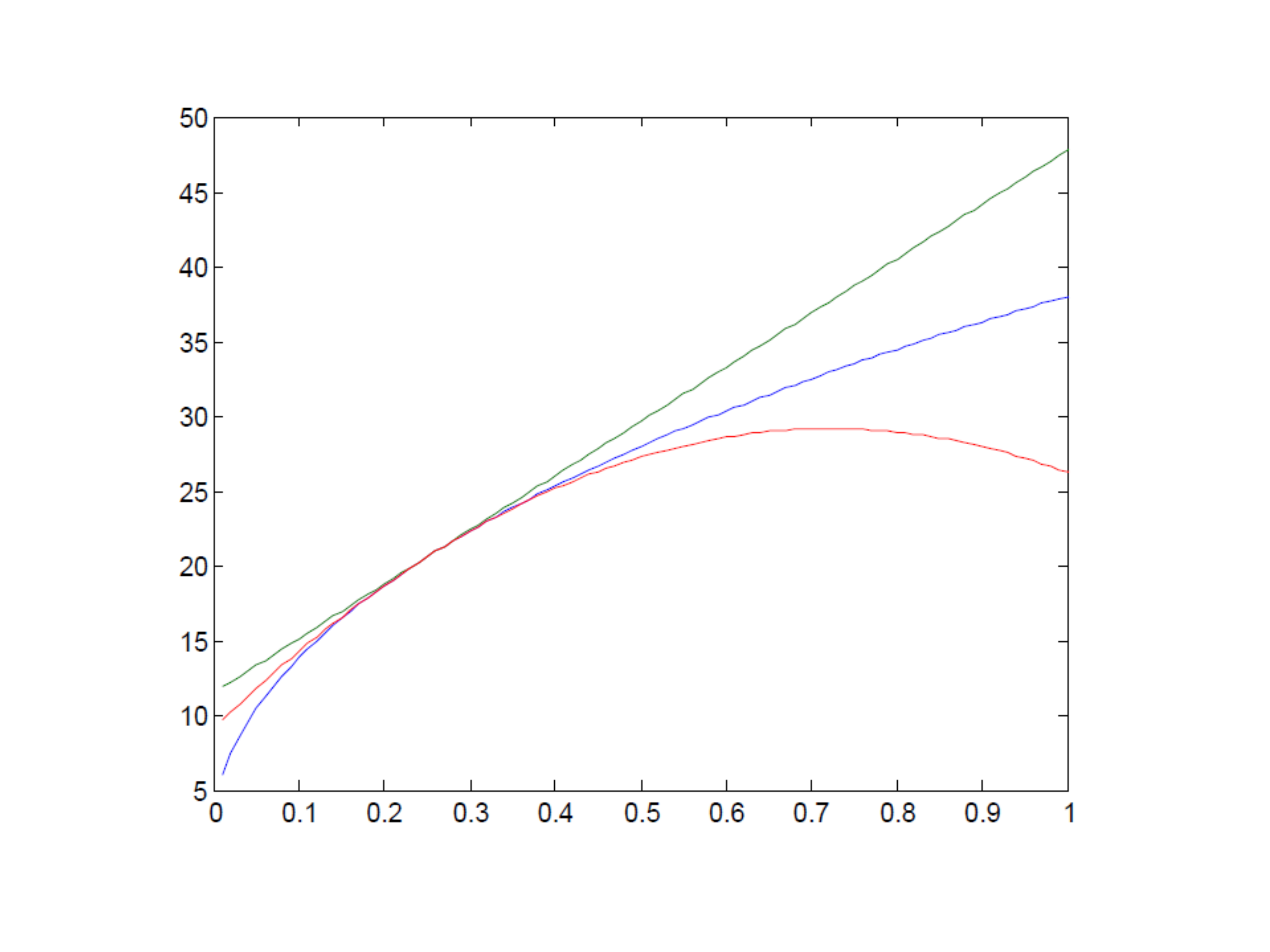}}}\hspace{5pt}
\subfigure[]{
\resizebox*{5cm}{!}{\includegraphics{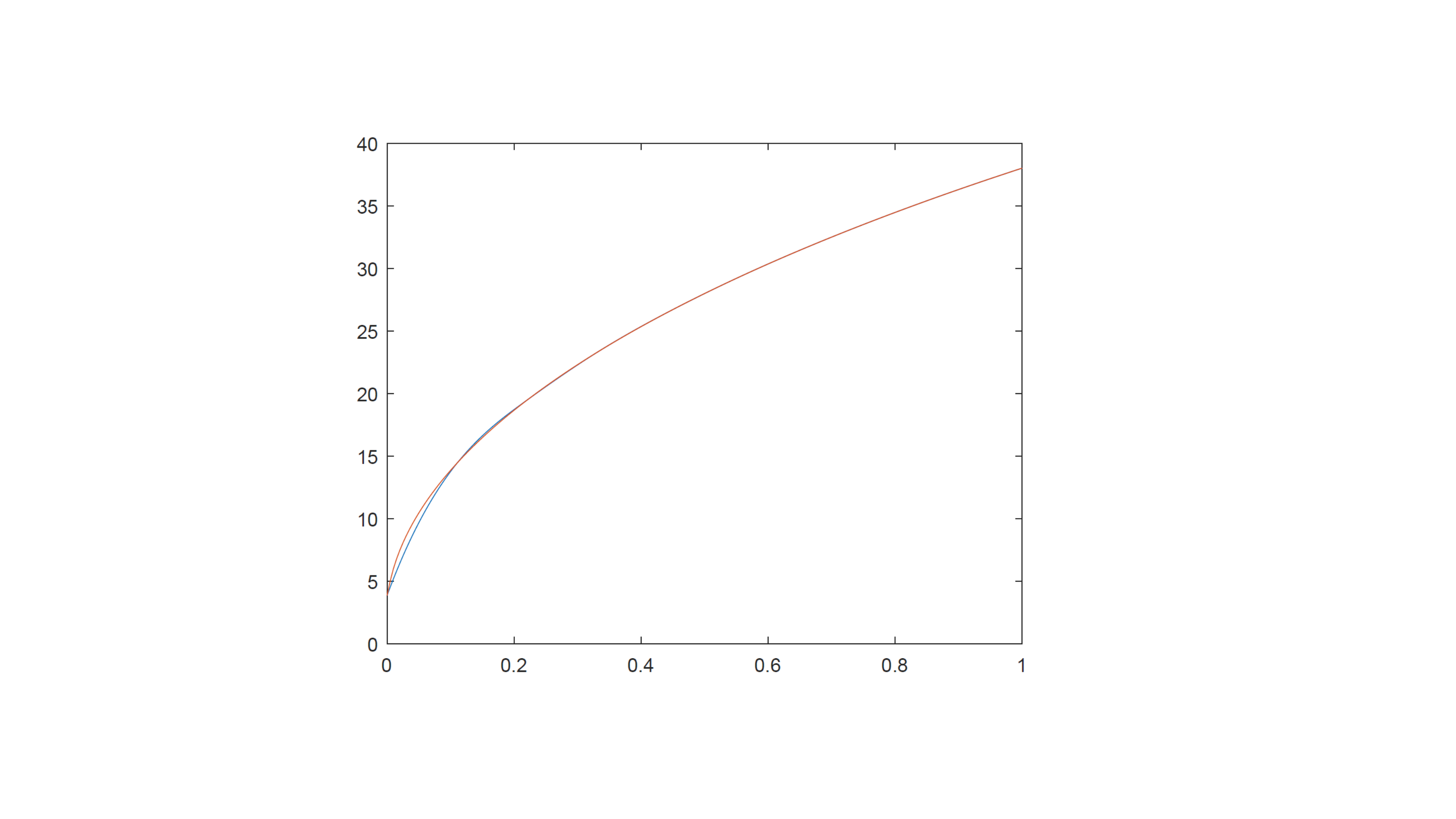}}}
\subfigure[]{
\resizebox*{5cm}{!}{\includegraphics{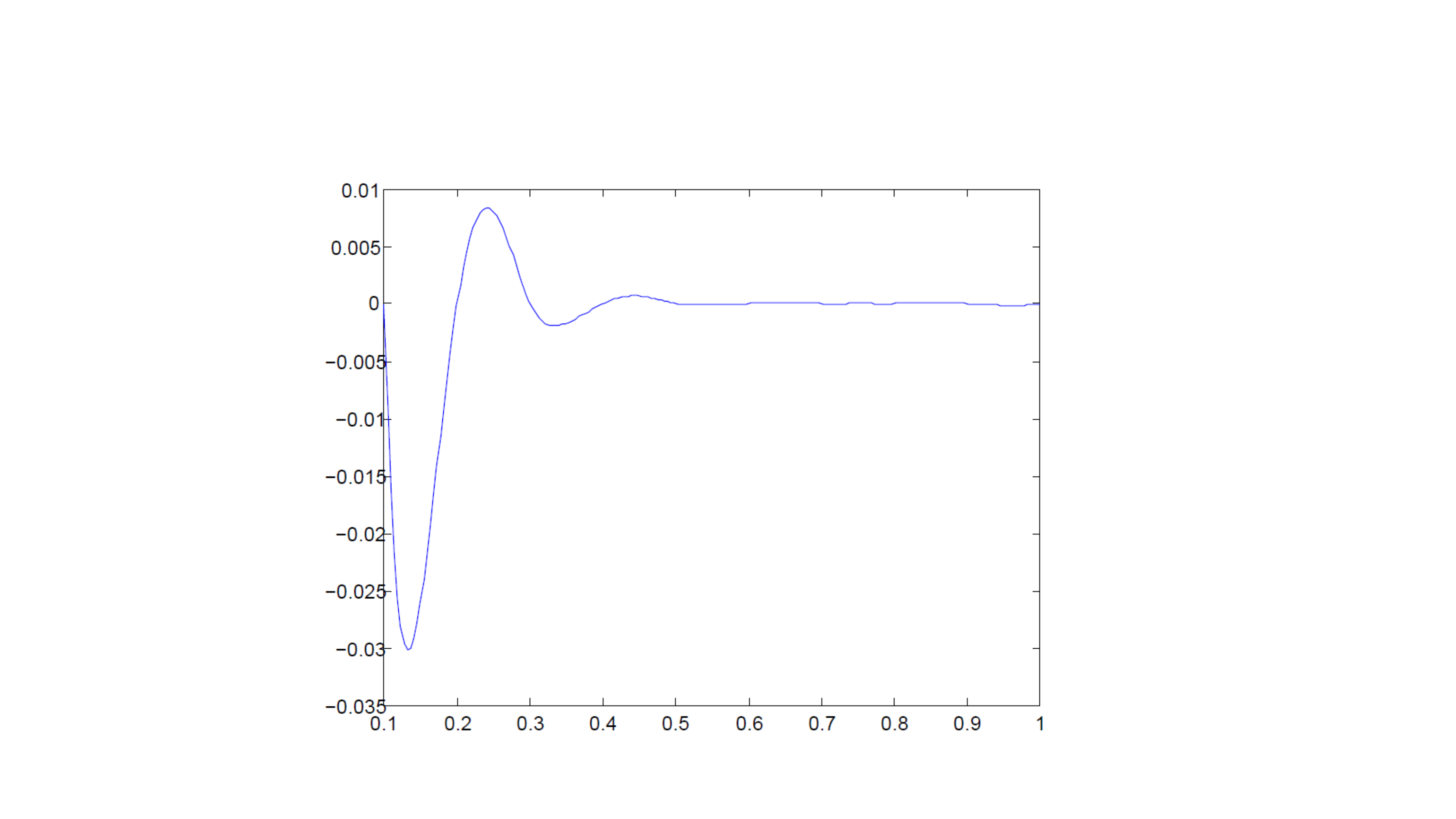}}}\hspace{5pt}
\subfigure[]{
\resizebox*{5cm}{!}{\includegraphics{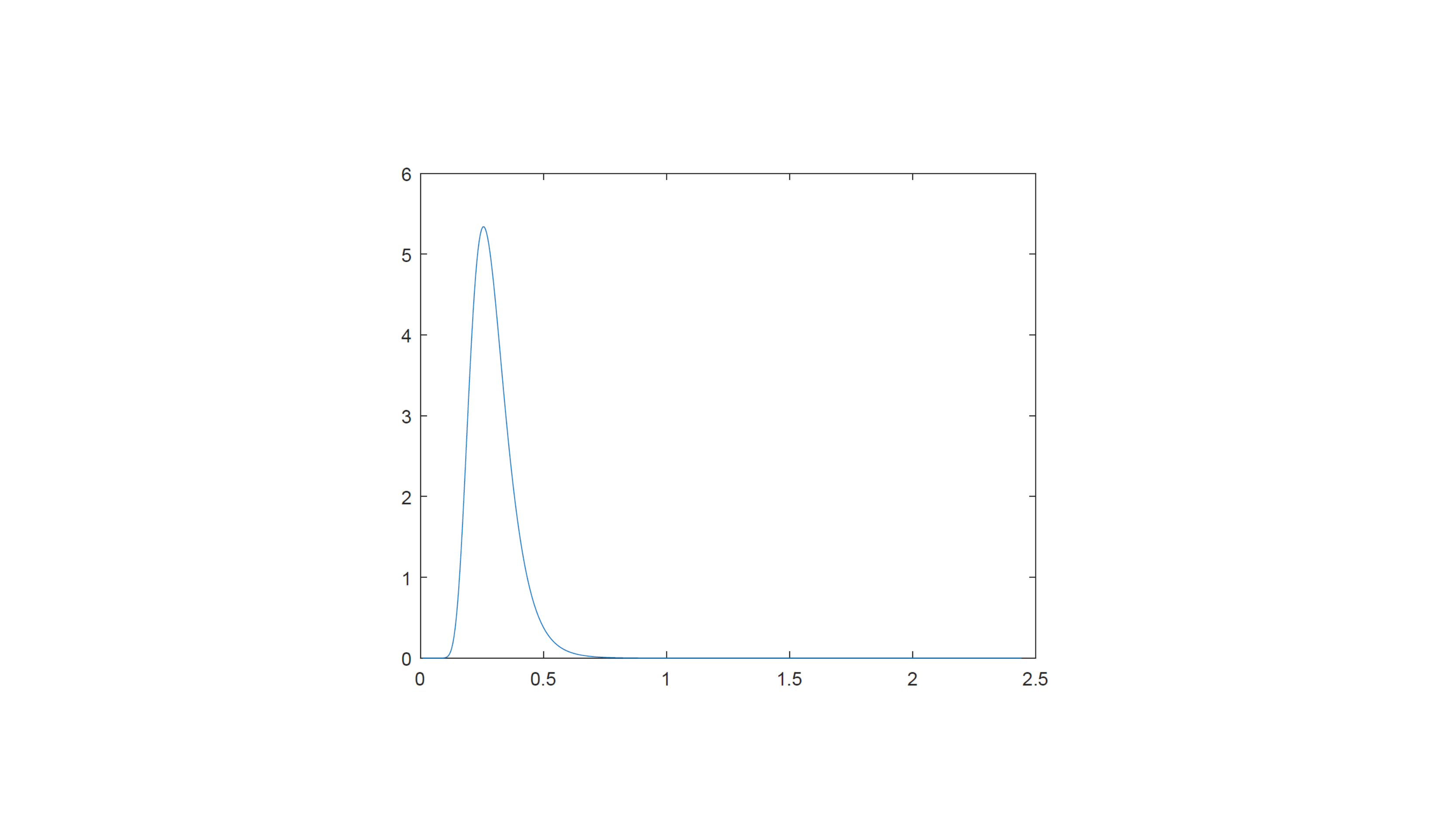}}}
\caption{  Figure (a): Margrabe prices as function of the parameter $v$ and its Taylor developments of first and second order around $v_0=0.25$. Figure (b):  Margrabe prices and its cubic splines approximation. Figure (c): Difference between Margrabe prices and its cubic spline approximation.  Figure (d): Empirical probability density prices  of  $v^+_T$ obtained after $10^5$ simulations}
\label{fig:taylorvsmc}
\end{center}
\end{figure}
Then, substituting equation (\ref{eq:mom2}) into equation (\ref{eq:taylorapproxuni}):
\begin{eqnarray*}
   \hat{C}^{(2)}_{MS}(v^*) & =& \hat{C}^{(1)}_{MS}(v^*)+\frac{1}{2}D^{2}C_{MT}(v^*)(v^*)^2\\
  &-& v^* D^{2}C_{MT}(v^*) \left[E_{\mathcal{Q}}(\sigma^{11+}_T)+E_{\mathcal{Q}}(\sigma^{22+}_T)-2  E_{\mathcal{Q}}(\sigma^{12+}_T) \right]\\
  &+& \frac{1}{2}D^{2}C_{MT}(v^*) \left[E_{\mathcal{Q}}(\sigma^{11+}_T)^2+2 E_{\mathcal{Q}}(\sigma^{11+}_T\sigma^{22+}_T) \right.\\
  &-&  4  E_{\mathcal{Q}}(\sigma^{11+}_T \sigma^{12+}_T)+ E_{\mathcal{Q}}(\sigma^{22+}_T)^2-  4  E_{\mathcal{Q}}(\sigma^{22+}_T \sigma^{12+}_T)\\
   &+& \left. 4   E_{\mathcal{Q}}(\sigma^{12+}_T)^2  \right]
  \end{eqnarray*}
  In figure \ref{fig:taylorvsmc}a) we show Margrabe price values as function of the variable $v$ (blue curve) on the interval $(0,1]$. For comparison, we also show Taylor polynomials of first (green line) and second  (red line) order around the  average log-price $v_0=0.25$ and  benchmark parameters specified in section 5. Both approximations are locally accurate but, for values farther from $v_0$, the  differences are shown to be significant.  It brings us the question of how often and how far departures from the average value occur?\\
    In figure \ref{fig:taylorvsmc}d) the pdf  of the random variable $v^+_T$ from  $10^5$ simulated values of $v$ is shown. It is estimated using a non-parametric Gaussian kernel. We observe that most values concentrate around the expansion point,  whereas a low but significant frequency appear far from the mean, indicating the presence of a  heavy-tailed probability distribution with positive skewness.\\
   In order to overcome this potential inconvenient we consider a cubic splines approximation. The latter adapts the expansion to the price behavior on different subintervals of $[a,b)$.\\
To compute the constrained moments of $v^+_T$ we use  proposition \ref{prop:constchf3}, equation (\ref{eq:dercnschf}). In order to simplify we assume initial values of the subordinators equal to zero.  \\
Therefore, we find that:
  \begin{eqnarray*}\nonumber
  m_{v^+_T}(0,a,b)&=& - \frac{i}{2 \pi} \int_{\mathbb{R}} f(y,a,b) g(-y)\;dy\\
  m_{v^+_T}(1,a,b)&=&  - \frac{1}{2 \pi} \int_{\mathbb{R}} f(y,a,b) ( D K^+_1(-My)+D K^+_2(-My))  g(-y)\;dy \\ \nonumber
   m_{v^+_T}(2,a,b)&=&   \frac{i}{2 \pi } \left[\int_{\mathbb{R}} f(y,a,b) ( D K^+_1(-My)+D K^+_2(-My))^2  g(-y)\;dy \right.\\ \nonumber
   &+& \left. \int_{\mathbb{R}} f(y,a,b) ( D^2 K^+_1(-My)+D^2 K^+_2(-My))  g(-y)\;dy \right]\\
    m_{v^+_T}(3,a,b)&=&   \frac{1}{2 \pi } \left[\int_{\mathbb{R}} f(y,a,b) ( D K^+_1(-My)+D K^+_2(-My)))^3  g(-y)\;dy \right.\\ \nonumber
   &+& 3 \int_{\mathbb{R}} f(y,a,b) ( D^2 K^+_1(-My)+D^2 K^+_2(-My))( D K^+_1(-My)+D K^+_2(-My)))  g(-y)\;dy\\
  &+& \left. \int_{\mathbb{R}} f(y,a,b) ( D^3 K^+_1(-My)+D^3 K^+_2(-My)) g(-y)\;dy \right]
  \end{eqnarray*}
  Higher moments are computed by recurrence:
  \begin{eqnarray*}
  m_{v^+_T}(k,a,b)&=& - \frac{i^{-k+1}}{2 \pi}\sum_{j=0}^{k-1}
    \left(\begin{array}{c}
    k-1\\ \label{eq:momconst2}
    j
    \end{array} \right)\\ \nonumber
    && \int_{\mathbb{R}}f(y,a,b)
 (D^{j+1}K^+_1(-My)+D^{j+1} K^+_2(-My)) D^{k-j-1}g(-y)\;dy \\
 && \text{for}\;\; k=2,3,\ldots
\end{eqnarray*}
Finally centered moments $\tilde{m}_{v^+_T}(k,a,b)$ are found from:
\begin{eqnarray*}
 \tilde{m}_{v^+_T}(k,a,b)&=& D^k \varphi_{v^+_T}(u)|_{u=0} \sum_{j=0}^{k}
    \left(\begin{array}{c}
    k\\
    j
    \end{array} \right) (-i a)^{k-j}m_{v^+_T}(j,a,b) 
\end{eqnarray*}
Some preliminary calculations of the functions $K_1$, $K_2$, $\varphi_{v^+_T}$ and their derivatives are shown in the appendix.\\
Alternatively, the constrained moments can be directly calculated from the pdf of $v^+_T$. In turn,  the pdf of $v^+_T$ is computed via its characteristic function by inverse FFT. To this end we define the grids:
\begin{eqnarray*}
  x_j &=& a+ \eta j,\;j=0,1,\ldots,n-1\\
  u_k &=& \delta k,\; k=0,1,\ldots,n-1
\end{eqnarray*}
where $\eta=\frac{b-a}{n}$ and $\delta=\frac{2 \pi}{b-a}$ are their respective lengths.\\
Hence, after applying the trapezoid rule:
\begin{eqnarray*}
  f_{v^+_T}(x_j) &=& \frac{1}{\pi} \int_0^{+\infty} Re(exp(-i x_j u)\varphi_{v^+_T}(u))\;du\\
  &\simeq&  \frac{1}{\pi} \sum_{k=0}^{n-1} w_k Re(exp(-ix_j u_k)\varphi_{v^+_T}(u_k)) \Delta u_k\\
  &=& \frac{1}{\pi} Re(\sum_{k=0}^{n-1} w_k \delta exp(-ia \delta k)\varphi_{v^+_T}(\delta k) exp(-i\frac{2 \pi}{n}jk))=fft(h_k)\\
\end{eqnarray*}
with $h_k= w_k \delta exp(-ia \delta k)\varphi_{v^+_T}(\delta k) $ and $w_0=w_{n-1}=\frac{1}{2}$ and equal to one otherwise. The expression $fft(h_k)$ denotes the Fast fourier Transform of the sequence $(h_k)$. \\
See \cite{wikt} for FFT applications in obtaining pdf's and \cite{hurw} for a detailed analysis of different quadratures.
\section{Numerical Results}
 We compare  the polynomial methods and the Monte Carlo approach to  pricing, for speed and accuracy. Our benchmark setting is given by a set of parameter values defining the model and the exchange contract. Contract parameters are selected  within a reasonable range, according to  usual practices, while the choosing  the model parameters is made rather arbitrary, just  with the purpose of illustrating the techniques.\\
 The benchmark parameters for the model are $a_{F}=(1, 1),  a_V=(1, 1),  b_F=(5, 5), b_V=(5, 5), \lambda_F=(1, 1), \lambda_V=(1, 1)$ and $S_0=(100,96)$. For the contract we set $ c=1, m=1, q=(0,0)$ and $T=1$. The interest rate is $r=0.04$.\\
 We take the loading matrix $A$ as an orthonormal rotation matrix with an angle $\theta, -\pi < \theta \leq \pi$, given by:
 \begin{eqnarray*}
A=\left(
\begin{array}{cc}
\cos \theta  & -\sin \theta  \\
\sin \theta  & \cos \theta %
\end{array}%
\right)
\end{eqnarray*}
 A direct Monte Carlo approach is costly as trajectories for both, the covariance process and the asset process, need to be simulated a large number of times. Alternatively, the iterative formula (\ref{eq:pricepi}) can be used to simplify calculations as, according to lemma \ref{returndistb}, conditionally on the covariance process the log-prices are normally distributed. It reduces the problem to calculate the discounted average of the price of an exchange contract under a deterministic time-dependent covariance, which still has a closed-form expression given in equation (\ref{eq:condprice}). Hence, only the Ornstein-Ulenbeck covariance process needs to be simulated. We call this procedure  a \textit{partial Monte Carlo} approach. \\
 Integrated Ornstein-Ulenbeck  process values at time $T$, denoted by  $\hat{F}_T^{l,+}$ and $\hat{V}_T^{l,+}$ are computed as discrete approximations of solutions of equations (\ref{eq:intouf}) and (\ref{eq:intouv}) with a step $\delta$, given respectively by:
        \begin{equation*}\label{eq:intoufapp}
    \hat{F}_T^{(l,+)}=\lambda_{F,l}^{-1} \left[(1-exp(-\lambda_{F,l}T)) F_0^{(l)}+ \sum_{k=1}^{n_1 } (1-exp(-\lambda_{F,l}(T-k \delta))) \Delta Z^{F,l}_{k} \right]
\end{equation*}
\begin{equation*}\label{eq:intouvapp}
    \hat{V}_T^{(l,+)}=\lambda_{V,l}^{-1} \left[(1-exp(-\lambda_{V,l}T)) V_0^{(l)}+ \sum_{k=1}^{n_2 } (1-exp(-\lambda_{V,l}(T-k \delta))) \Delta Z^{V,l}_{k} \right]
\end{equation*}
where:\\
$n_1=\left[ \frac{\lambda_{F,l} T}{\delta}\right]$ and $\Delta Z^{F,l}_{k}=Z^{F,l}_{k\delta}-Z^{F,l}_{(k-1)\delta}$,\\
 $n_2=\left[ \frac{\lambda_{V,l} T}{\delta}\right]$ and $\Delta Z^{V,l}_{k}=Z^{V,l}_{k\delta}-Z^{V,l}_{(k-1)\delta}$, $l=1,2$.\\
 The symbol $[x]$ is the integer part of the real value $x$.\\
  Next, the integrated covariance process is computed:
     \begin{equation*}\label{eq:covfactorpc2}
     \Sigma^+_T= diag( F^+_T)+A \; diag(V^+_T) A'
\end{equation*}
The price of the derivative contract is estimated from  equation (\ref{eq:pricepi}) by  the simulation of the covariance process and then computing the discounted average of the Margrabe prices evaluated at these simulated volatilities.\\
 As an illustration, in figure \ref{fig:ig}a) three trajectories of an Inverse Gaussian process $(Z_t)_{0 \leq t \leq 1}$ with parameters $a=1$, $b=5$  are shown. Next, we generate the corresponding Ornstein-Uhlenbeck process $(F_t)_{0 \leq t \leq 1}$ as shown in figure \ref{fig:ig}b), starting at zero.
\begin{figure}
\begin{center}
\subfigure[Three realizations of an Inverse Gaussian process with parameters $a=1, b=5$.]{
\resizebox*{5cm}{!}{\includegraphics{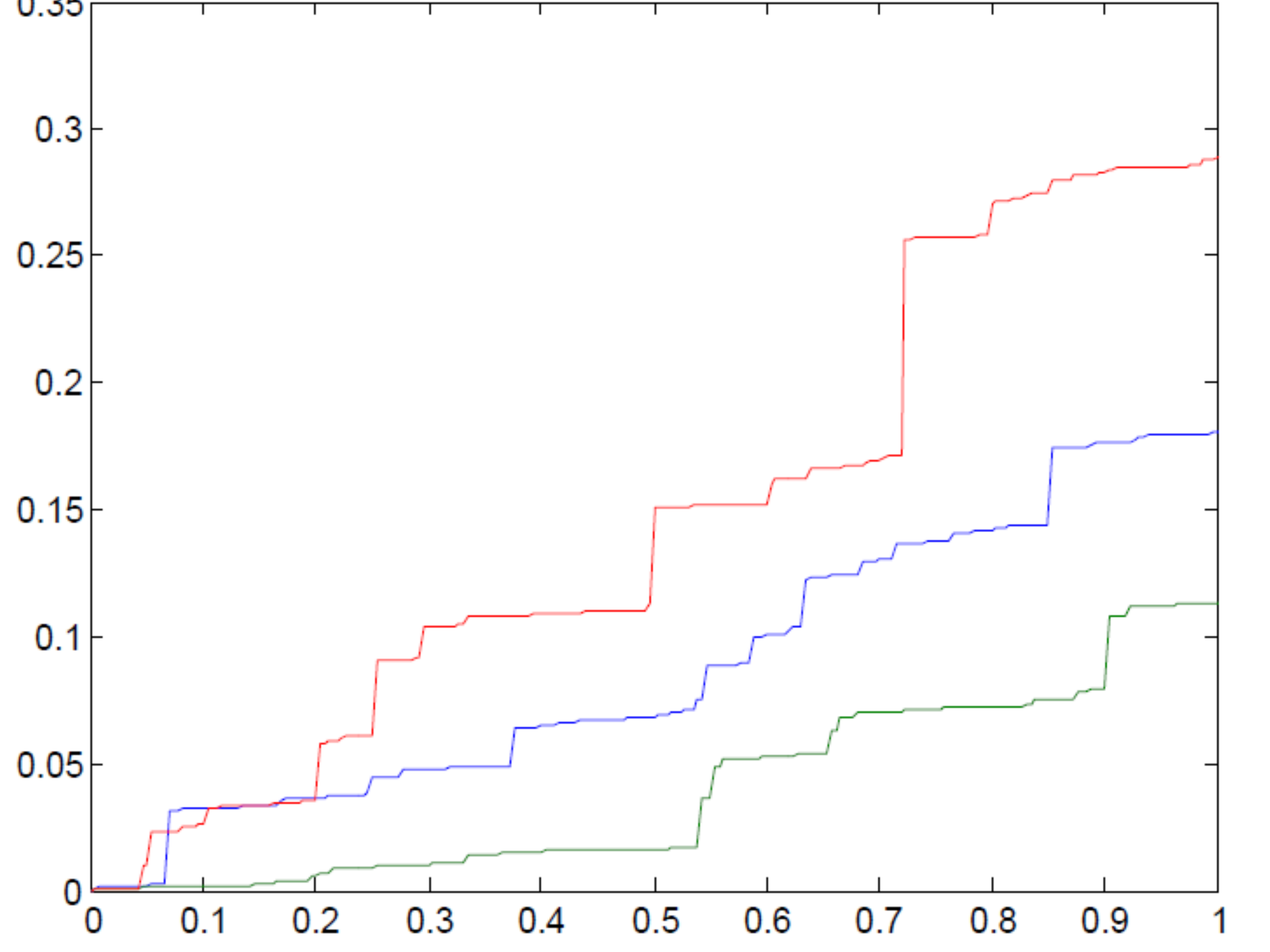}}}\hspace{5pt}
\subfigure[Three realizations of an Ornstein-Uhlenbeck  process with parameters $a=1, b=5$ and $\lambda=1$.]{
\resizebox*{6cm}{!}{\includegraphics{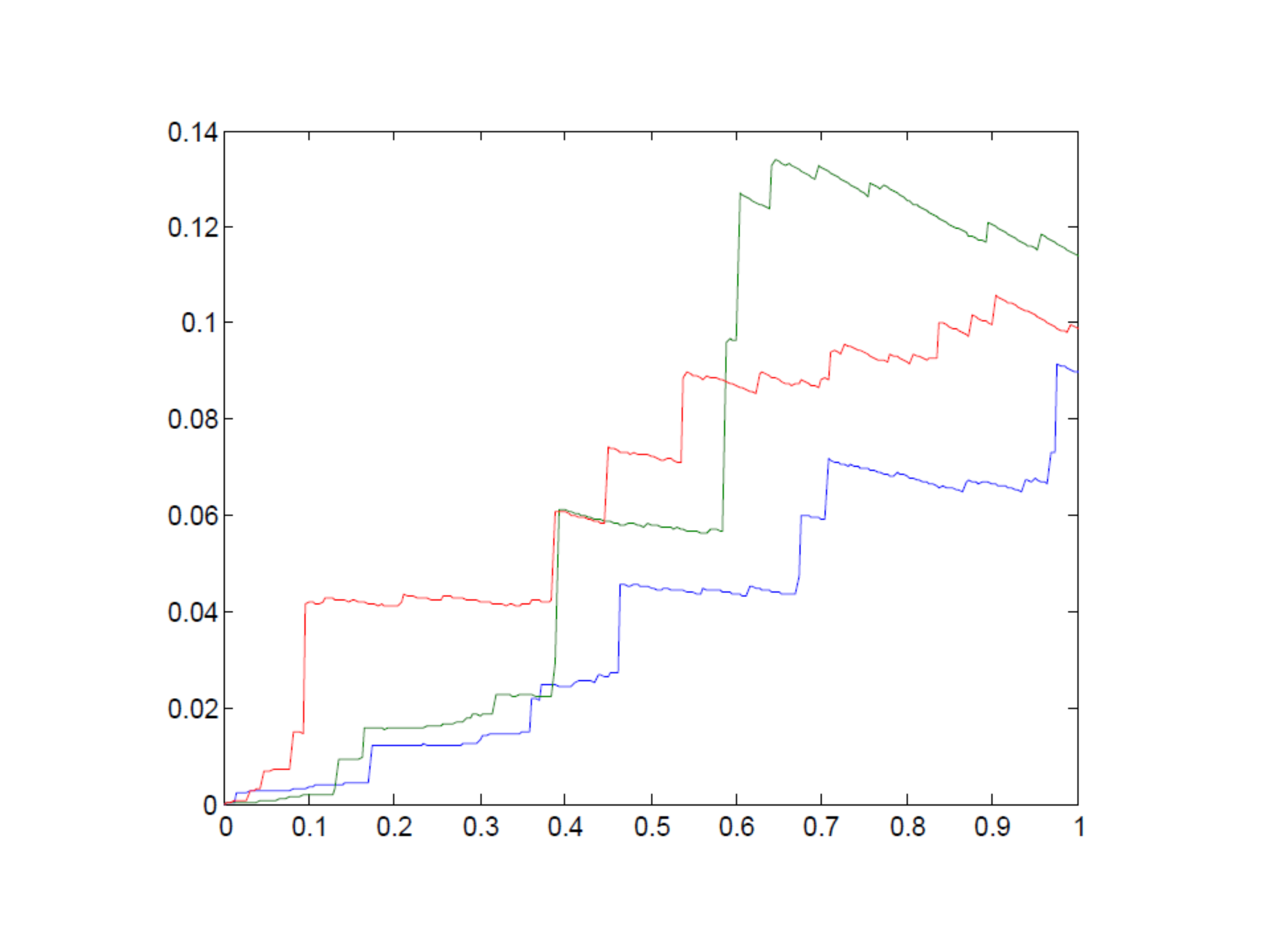}}}
\subfigure[Three realizations of the correlation process]{
\resizebox*{6cm}{!}{\includegraphics{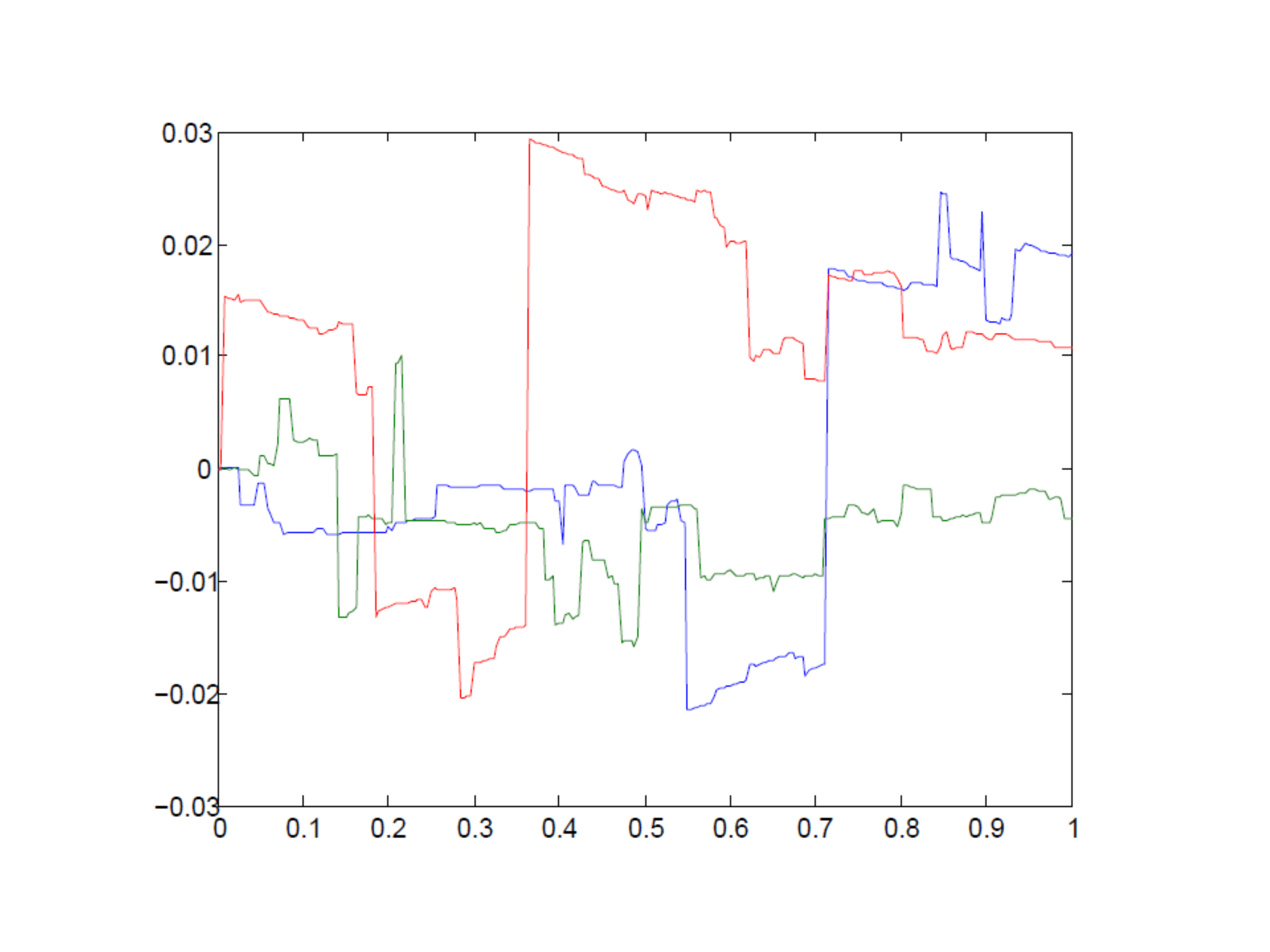}}}\hspace{5pt}
\caption{ Trajectories of the correlation process for selected parameters}
\label{fig:ig}
\end{center}
\end{figure}
Finally, in figure \ref{fig:ig}c), we show the trajectories of the correlation process obtained by dividing the covariance process $(\sigma^{12}_t)_{0 \leq t \leq 1}$ by the product of volatilities from the underlying assets, with a load matrix defined by the angle $\theta=\frac{\pi}{6}$. Both processes $(F_t)_{0 \leq t \leq 1}$ and $(V_t)_{0 \leq t \leq 1}$ are generated with the benchmark model parameters. Notice the correlation process exhibits jumps at random times, accounting for unexpected events.\\
In table \ref{tab:prices} different prices of the exchange contract  for some notable values of the angle in the loading matrix are shown. In the case of the Monte Carlo approach we also calculate a 95\% confidence interval for the price after 1 million simulations. We see that all methods, except the second order Taylor expansion, are within a similar range. First order Taylor price presents  inaccuracies for other parameters of the subordinator processes, while the ones based on cubic splines and FFT developments are quite stable, their relative average error are approximately 0.018 \% when compared with the estimated Monte Carlo price.\\
On the other hand, in table \ref{tab:running} we can see the execution time (in sec.) for all five methods. The code was written on a surface pro 4 i7 using MATLAB language. Cubic splines and FFT methods are, on average, respectively $16488.8$ and $18408.4$ times faster than Monte Carlo. The fact that FFT is slightly faster than cubic splines approximation comes at no surprise. It is well known that the former has a $0(n log n)$ complexity compared with a $0(n^2)$ of the latter.\\
In implementing both approaches some set of parameters driving the numerical approximations are required.  Namely, it is needed to decide on the truncation interval $[a,b)$, to fix a number of points in the grid for the Fourier transform and the number of points in the spline interpolation. The three factors require a compromise between accuracy and the amount of computation time. In the choice of the truncation interval we have tried to cover most of the support of the pdf of the random variable $v^+_T$ (see figure \ref{fig:taylorvsmc}d))which in turn depends on the parameters of both Inverse Gaussian subordinator processes. In our setting the interval $[0,5)$ was a reasonable tradeoff. Of course, most of the time these parameters need to be estimated. In any case a significant probability mass is present in a neighborhood of zero, therefore $a=0$ seems an evident choice.\\
For the number of points in the grid of the FFT calculation we have tested several powers of 2, ranging from $2^8$ to $2^{14}$. There is not a significant change in the calculated price across these values. We have set an intermediate value of $2^{12}$. After this figure the computation time explodes without a significant gain in accuracy. In order to implement the approach based on spline polynomials, we explored a range from $2^4$ to $2^8$ of interpolation points. After $2^6$ the price values are in close agreement with Monte Carlo and FFT. Numerical results improve if first derivatives at the end points of the expansion intervals are taken into account. They are available via formula \ref{eq:dercmt}.\\
Notice that the FFT approach refers to the manner the pdf of the stochastic volatility is obtained. It differs from standard FFT techniques based on the Fourier transform  of the payoff, see \cite{carr}.
\begin{table}[htp!]
  \centering
  \begin{tabular}{|c|c|c|c|c|c|}
    \hline
    $\theta$ & MC & Taylor (first order) & Taylor (sec.  order) &  Spl &FFT   \\ \hline
     $\frac{\pi}{6}$ & 21.8643 &22.1774   & 19.8441& 21.8969 & 21.8990\\
     &$(21.8589,21.8697)$&&&& \\ \hline
     $\frac{\pi}{3}$ & 21.8610& 22.1773 &  19.8441 &21.8969 & 21.8990\\
     &$(21.8557, 21.8664)$ & & & & \\ \hline
         $\frac{\pi}{2}$  & 21.9191 & 22.1773  &20.6861&  21.9506     &  21.9521 \\
       &$(21.9144, 21.9238)$& & & &  \\  \hline
       $\pi$ &21.9163 &22.1774  & 20.6861 &21.9506 & 21.9521\\
       & $(21.9117, 21.9210)$& & & & \\ \hline
     \end{tabular}
     	  \caption{Prices obtained from different approximations and different values of $\theta$ using the benchmark parameters}\label{tab:prices}
\end{table}

 \begin{table}[h!]
  \centering
  \begin{tabular}{|c|c|c|c|c|}
    \hline
     & MC &Taylor & Spl & FFT   \\ \hline
   Run time  & 179.6294 &0.010847  &0.010894 &0.009758   \\ \hline
  \end{tabular}\label{tab:running}
  \caption{Computer time (in seconds) for different pricing methods using the benchmark parameters and $\theta=\frac{\pi}{6}$}
\end{table}

\section{Conclusions}
We have discussed the pricing of exchange contracts under a dynamic model with a complex correlation structure capturing random jumps, heavy-tails, asymmetric and stochastic behavior. To this end we have proposed two approximate closed-form pricing methods based on cubic splines approximation and inverse Fast Fourier Transform. To the extend of the investigation and the range of parameters considered, in the model, the contract and the numerical method, both approaches provide accurate and fast pricing estimations.
\section{Acknowledgments}
This research was partially supported by the Natural Sciences and Engineering Research Council of Canada.
\section{Appendix}
\textbf{Proof of theorem \ref{chfbnmodel}}\\
From equations (\ref{eq:intouf}) and (\ref{eq:intouv}), and  Levy-Khinchine formula  we have, for $\tilde{\theta} \in \mathbb{R}$:
\begin{eqnarray*}
\varphi_{ F_t^{(l,+)}}(\tilde{\theta}) &=&  exp(i  \tilde{\theta} \lambda_{F,l}^{-1}(1-exp(-\lambda_{F,l}t))F_0^{(l)})E_{\mathcal{Q}} \left[ exp(i \tilde{\theta} \lambda_{F,l}^{-1} \int_0^t(1-exp(-\lambda_{F,l}(t-s)))dZ^{F,l}_{\lambda_{F,l}s}) \right]\\
&=& exp(i \tilde{\theta} \lambda_{F,l}^{-1}(1-exp(-\lambda_{F,l}t))F_0^{(l)}+\int_0^{\lambda_{F,l}t}\Psi_{Z^{F,l}}(\lambda_{F,l}^{-1}\tilde{\theta}(1-exp(-\lambda_{F,l}t+s)))ds)\\
\end{eqnarray*}
where the last equality follows from the identity:
\begin{equation*}
  E_{\mathcal{Q}} \left(exp(i \int_0^t f(s) dL_s) \right)=exp(\int_0^t \Psi_L(f(s))\;ds)
\end{equation*}
for a Levy process $L$ with characteristic exponent $\Psi_L$ and a continuous function  $f$, see for example \cite{cont}.
Similarly:
\begin{eqnarray*}
\varphi_{ V_t^{(l,+)}}(tr( \theta A C_{l}A'))&=& exp(i tr( \theta A C_{l}A') \lambda_{V,l}^{-1}(1-exp(-\lambda_{V,l}t))V_0^{(l)}\\
&+& \int_0^{\lambda_{V,l}t}\Psi_{Z^{V,l}}(tr(\theta A C_l A')\lambda_{V,l}^{-1}(1-exp(-\lambda_{V,l}t+s)))ds)
\end{eqnarray*}
Next, the characteristic function for the integrated covariance process can be computed as:
\begin{eqnarray*} \notag
\varphi_{\Sigma^+_t}(\theta)&=& E_{\mathcal{Q}} exp(  i~tr(\theta diag(F^+_t)))  E_{\mathcal{Q}} exp(  i~tr(\theta A\; diag(V^+_t) A')) \\ \notag
&=& exp(i  \sum_{l=1}^2 \lambda_{F,l}^{-1}(1-exp(-\lambda_{F,l}t))\theta_{ll}F_0^{(l)}+\int_0^{\lambda_{F,l}t}\Psi_{Z^{F,l}}(\lambda_{F,l}^{-1}\theta_{ll}
(1-exp(-\lambda_{F,l}t+s)))ds)\\
&& exp(i \sum_{l=1}^d tr(  \theta AC_{l}A') \lambda_{V,l}^{-1}(1-exp(-\lambda_{V,l}t))\\
&+& \sum_{l=1}^2  \int_0^{\lambda_{V,l}t}\Psi_{Z^{V,l}}(\lambda_{V,l}^{-1} tr(  \theta AC_{l}A')(1-exp(-\lambda_{V,l}t+s)))ds)\\
 \end{eqnarray*}
On the other hand, from equation (\ref{eq:charexpf}) we have:
\begin{eqnarray*}
\Psi_{Z^{F,l}}(\tilde{\theta} \lambda_{F,l}^{-1}(1-exp(-\lambda_{F,l}t+s)))&=&- \frac{a_{F,l}}{\sqrt{i \lambda_{F,l}}} \sqrt{2\, \tilde{\theta} \left( 1-exp(-\lambda_{F,l}\,t+s) \right) +i \lambda_{F,l}{b^{2}_{F,l}}}+a_{F,l}b_{F,l}
\end{eqnarray*}
Then, for $\tilde{\theta} \neq 0$:
\begin{eqnarray} \notag
&& I_F^{(l)}(\lambda_{F,l} t,\tilde{\theta})= -\frac{a_{F,l}}{\sqrt{ i \lambda_{F,l}}}  \int_0^{\lambda_{F,l}t} \sqrt{2\, \tilde{\theta} \left( 1-exp(-\lambda_{F,l}\,t+s)  \right) +i \lambda_{F,l}b^{2}_{F,l}} ds+\lambda_{F,l}a_{F,l}b_{F,l}t\\ \notag
&=& -{\frac {2\,a_{F,l}}{\sqrt { i \lambda_{F,l}}}}\left[ - T^{F,l}_2(\tilde{\theta})+T^{F,l}_1(\tilde{\theta})\arctan \left( \frac{T^{F,l}_2(\tilde{\theta})}{T^{F,l}_1(\tilde{\theta})}\right)\right. \\ \label{eq:zzz}
&+& \left.\sqrt{i \lambda_{F,l}}b_{F,l}-T^{F,l}_1(\tilde{\theta})\arctan \left( \frac{\sqrt{i \lambda_{F,l}}b_{F,l}}{T^{F,l}_1(\tilde{\theta})}\right)\right]+\lambda_{F,l} a_{F,l} b_{F,l} t
  \end{eqnarray}
  In what follows we use the identity:
  \begin{equation*}
    arctan z=-\frac{i}{2} \log \left( \frac{(1+iz)^2}{1+z^2} \right)\;\;z \in \mathbb{C}\\ \{ -i,i \}
  \end{equation*}
  where the complex-valued logarithmic function is defined according to  the principal value of the argument.\\
  Notice that:
  \begin{eqnarray*}
    \frac {\sqrt {i\lambda_{F,l}}b_{F,l}}{T^{F,l}_1(\tilde{\theta})}   &=& i \Longleftrightarrow   \tilde{\theta}=0\\
 \end{eqnarray*}
 Similarly:
  \begin{eqnarray*}
  \frac{T^{F,l}_2(\tilde{\theta})}{T^{F,l}_1(\tilde{\theta})} &=& i \Longleftrightarrow  2 \tilde{\theta} \, \left( 1-{{\rm e}^{-\lambda_{F,l}\,t}} \right)+i\lambda_{F,l} \,b^{2}_{F,l}=-(-2\,\tilde{\theta}-i\lambda_{F,l}\,b^{2}_{F,l})
   \end{eqnarray*}
 whose solution is again $\tilde{\theta} = 0$.\\
  An analysis at $-i$ leads to the same conclusion.\\
  Hence, for $\tilde{\theta} \neq 0$ we have:
  \begin{eqnarray*}
\arctan \left(\frac {\sqrt {i \lambda_{F,l}}b_{F,l}}{T^{F,l}_1(\tilde{\theta})}\right)&=&-\frac{i}{2}\log \left[ \frac{\left(1+\frac{i \sqrt{i \lambda_{F,l}}\,b_{F,l}}{T^{F,l}_1(\tilde{\theta})}\right)^2}{1+\left( \frac{ \sqrt{i \lambda_{F,l}}\,b_{F,l}}{T^{F,l}_1(\tilde{\theta})}\right)^2}\right]\\
&=& -\frac{i}{2} \left[ 2 \log \left(T^{F,l}_1(\tilde{\theta})+i  \sqrt{i \lambda_{F,l}}\,b_{F,l}\right)-\log ( -2 \tilde{\theta})\right]
\end{eqnarray*}
and
\begin{eqnarray*}
&&\arctan \left( {\frac {T^{F,l}_2(\tilde{\theta})}
{  T^{F,l}_1(\tilde{\theta})}} \right) =-\frac{i}{2} \left[ 2 \log \left(T^{F,l}_1(\tilde{\theta})+iT^{F,l}_2(\tilde{\theta}) \right)-\log ( -2 \tilde{\theta}exp(- \lambda t)) \right]
\end{eqnarray*}
Then, substituting the expressions  for arctan above into equation (\ref{eq:zzz}) we obtain equation (\ref{eq:igintf}) after noticing that:
\begin{eqnarray*}\label{eq:ifoulevyf}
 I_F^{(l)}(\lambda_{F,l} t,\tilde{\theta}) &=& {-\frac {2\,a_{F,l}}{\sqrt {i \lambda_{F,l}}}}\left[ -T^{F,l}_2(\tilde{\theta})+\sqrt{i \lambda_{F,l}}b_{F,l}\right.\\
 &+& \left.\frac{i}{2}T^{F,l}_1(\tilde{\theta}) \log \left( \frac{(T^{F,l}_1(\tilde{\theta})+i \sqrt{i \lambda_{F,l}}\,b_{F,l})^2( -2 \tilde{\theta}exp(- \lambda_{F,l} t))}{(T^{F,l}_1(\tilde{\theta})+i T^{F,l}_2(\tilde{\theta}))^2( -2 \tilde{\theta})}\right)\right]\\
 &+&\lambda_{F,l} a_{F,l} b_{F,l} t
 \end{eqnarray*}
  In a similar way we obtain equation (\ref{eq:igintv}). Notice that in the case of $I^{(l)}_V(\theta)$, the equality is valid for all values of the matrix $\theta$  except when $tr(\theta A C_l A')=0$, which is equivalent to $\theta_{ll}=0$.\\
\textbf{Proof of proposition 3.2}\\
The proof is straightforward. It is based on computing the first and second derivatives of the characteristic function evaluated at zero. \\
Hence, for the first moments we notice that:
\begin{eqnarray*}\label{eq:k1k}
  \frac{\partial  K^+_1(\theta)}{\partial \theta_{kk}} &=& i\lambda_{F,k}^{-1}(1-exp(-\lambda_{F,k}t))F_0^{(k)}+\frac{\partial  I_F^{(k)}(\lambda_{F,k} t, \theta_{kk})}{\partial \theta_{kk}} \\ \label{eq:k1kj}
  \frac{\partial  K^+_1(\theta)}{\partial \theta_{kj}}&=&0,\;\; k \neq j\\ \nonumber
   \frac{\partial  K^+_2(\theta)}{\partial \theta_{kj}}&=&i \sum_{l=1}^2 a_{kl}a_{jl} \lambda_{V,l}^{-1}(1-exp(-\lambda_{V,l}t))V_0^{(l)}\\
   &+&  \sum_{l=1}^2 a_{kl}a_{jl} \frac{\partial I_V^{(l)}(\lambda_{V,l}t,tr(\theta AC_{l}A'))}{\partial \theta_{kj}}  \label{eq:k2kj}
      \end{eqnarray*}
 Differentiating the characteristic function with respect to the components of the matrix $\theta$ and evaluating at $\theta=0$ we have:
  \begin{eqnarray*} \label{eq:moment1kk}
     E_{\mathcal{Q}}(\sigma^{kk+}_t) &=&  -i\left(\frac{\partial K^+_1(\theta) }{\partial \theta_{kk}}\mid_{\theta=0}+\frac{\partial K_2(\theta) }{\partial \theta_{kk}}\mid_{\theta=0} \right) \\ \nonumber
     && \\ \label{eq:moment1kj}
           E(\sigma^{12+}_t)&=&-i\frac{\partial K^+_2(\theta) }{\partial \theta_{12}}\mid_{\theta=0}
             \end{eqnarray*}
 From which,  equations (\ref{eq:moment1gralkkf}) and (\ref{eq:moment1gral12f}) follow.\\
On the other hand the second partial derivatives after evaluating at zero are:
\begin{eqnarray*}\label{eq:k1k0}
  \frac{\partial^2  K^+_1(\theta)}{\partial \theta_{kj}\partial \theta_{mn}}|_{\theta=0}&=& 0, \;\text{if at least one of subscripts is different}\\ \label{eq:k2kj0}
  \frac{\partial^2  K^+_1(\theta)}{\partial \theta^2_{kk}}|_{\tilde{\theta}=0}&=&\frac{\partial^2 I_F^{(l)}(\lambda_{F,l} t, \theta_{kk}) }{\partial \theta^2_{kk}}|_{\theta=0}\\
  \frac{\partial^2  K^+_2(\theta)}{\partial \theta_{kj} \partial \theta_{mn}}|_{\theta=0}&=&  \sum_{l=1}^2 a_{kl}a_{jl}a_{ml}a_{nl}\frac{\partial^2 I_V^{(l)}(\lambda_{V,l}t,tr(\theta AC_{l}A'))}{\partial \theta_{kj}\partial \theta_{mn}}|_{\theta=0}
   \end{eqnarray*}
 Therefore, we have:
   \begin{eqnarray*}\notag
      E_{\mathcal{Q}}(\sigma^{kj+}_t \sigma^{mn+}_t) &=& -\left(\frac{ \partial^2 K^+_1(\theta) }{\partial \theta_{kj} \partial \theta_{mn}} + \frac{\partial^2 K^+_2(\theta) }{\partial \theta_{kj} \partial \theta_{mn}}\right)\mid_{\theta=0}\\ \notag
    &-& \left(\frac{\partial K^+_1(\theta) }{\partial \theta_{kj}}+\frac{\partial K^+_2(\theta) }{\partial \theta_{kj}}\right)\mid_{\theta=0} \left(\frac{\partial K^+_1(\theta) }{\partial \theta_{mn}}+\frac{\partial K^+_2(\theta) }{\partial \theta_{mn}}\right)\mid_{\theta=0}\label{eq:crossmoments}
    \end{eqnarray*}
  In particular:
  \begin{eqnarray*}
      E_{\mathcal{Q}}(\sigma^{12+}_t)^2 &=&  - \frac{\partial^2 K^+_2(\theta) }{\partial \theta^2_{12}}\mid_{\theta=0}- \left( \frac{\partial K^+_2(\theta) }{\partial \theta_{12}}\right)^2\mid_{\theta=0}\\
      E_{\mathcal{Q}}(\sigma^{kk+}_t \sigma^{12+}_t) &=&   -\left(\frac{ \partial^2 K^+_1(\theta) }{\partial \theta_{kk} \partial \theta_{12}} + \frac{\partial^2 K^+_2(\theta) }{\partial \theta_{kk} \partial \theta_{12}}\right)\mid_{\theta=0}\\ \notag
    &-& \left(\frac{\partial K^+_1(\theta) }{\partial \theta_{kk}}+\frac{\partial K^+_2(\theta) }{\partial \theta_{kk}}\right)\mid_{\theta=0} \left(\frac{\partial K^+_1(\theta) }{\partial \theta_{12}}+\frac{\partial K^+_2(\theta) }{\partial \theta_{12}}\right)\mid_{\theta=0}\\
     &=& - \frac{\partial^2 K^+_2(\theta) }{\partial \theta_{kk} \partial \theta_{12}}\mid_{\theta=0}- \left(\frac{\partial K^+_1(\theta) }{\partial \theta_{kk}}+\frac{\partial K^+_2(\theta) }{\partial \theta_{kk}}\right)\mid_{\theta=0} \frac{\partial K^+_2(\theta) }{\partial \theta_{12}}\mid_{\theta=0}\label{eq:crossmoments2}
     \end{eqnarray*}
 Also:
   \begin{eqnarray*}
    E_{\mathcal{Q}}(\sigma^{11+}_t \sigma^{22+}_t) &=&-\frac{\partial^2 K_2(\theta) }{\partial \theta_{11} \partial \theta_{22}}\mid_{\theta=0}\\
    &-& \left(\frac{\partial K^+_1(\theta) }{\partial \theta_{11}}+\frac{\partial K^+_2(\theta) }{\partial \theta_{11}}\right)\mid_{\theta=0}\left(\frac{\partial K^+_1(\theta) }{\partial \theta_{22}}+\frac{\partial K^+_2(\theta) }{\partial \theta_{22}}\right)\mid_{\theta=0}
  \end{eqnarray*}
Finally:
 \begin{eqnarray*}
    E_{\mathcal{Q}}(\sigma^{kk+}_t)^2  &=& -\left(\frac{\partial K^+_1(\theta) }{\partial \theta_{kk}}\mid_{\theta=0}+\frac{\partial K^+_2(\theta) }{\partial \theta_{kk}}\mid_{\theta=0}\right)^2- \left(\frac{\partial^2 K^+_1(\theta) }{\partial \theta^2_{kk}}+\frac{\partial^2 K^+_2(\theta) }{\partial \theta^2_{kk}} \right)\mid_{\theta=0}
          \end{eqnarray*}
 \textbf{ Proof of proposition 3.3}\\
  Denote by $\hat{1}_{[a,b]}$ the Fourier transform of $1_{[a,b]}$.\\
 Notice that:
\begin{equation*}
  \hat{1}_{[a,b]}(y)=-i \frac{exp(iby)-exp(iay)}{y},\; y \neq 0
\end{equation*}
and equal to $b-a$ if $y=0$.\\
Hence:
\begin{eqnarray*}
 \varphi_{v^+_T}(u,a,b)  &=&  \int_{\mathbb{R}} exp(i u x) 1_{[a,b]}(x) \mathcal{Q}_{v^+_T}(dx)\\
  &=&  \frac{1}{2 \pi} \int_{\mathbb{R}} exp(i u x) \left[ \int_{\mathbb{R}} exp(-i y x) \hat{1}_{[a,b]}(y) \;dy \right] \mathcal{Q}_{v^+_T}(dx)\\
&=&  \frac{1}{2 \pi} \int_{\mathbb{R}}  \left[ \int_{\mathbb{R}} exp(i (u-y) x) \mathcal{Q}_{v^+_T}(dx) \right]\hat{1}_{[a,b]}(y) \;dy  \\
&=&  \frac{1}{2 \pi} \int_{\mathbb{R}} \varphi_{\Sigma^+_T}(M(u-y))\hat{1}_{[a,b]}(y) \;dy  \\
&=& - \frac{i}{2 \pi} \int_{\mathbb{R}} \varphi_{\Sigma^+_T}(M(u-y))f(y) \;dy
\end{eqnarray*}
Equation (\ref{eq:constchf2}) easily follows from theorem \ref{chfbnmodel}. The second part of the proposition follows from elementary differentiation.

          From which  equations (\ref{eq:moment2gralkkf2})-(\ref{eq:moment1122uneq3}) easily follow.\\
\textbf{Preliminary calculations for the constrained moments}\\
 Some preliminary calculations of the constrained moments are given below.
\begin{eqnarray*}
    T^{F,l}_1(-y)&=& \sqrt {2y-i\lambda_{F,l}\,{b^{2}_{F,l}}}\\
T^{F,l}_2(T,-y)&=&\sqrt {-2y(1-{{\rm e}^{-\lambda_{F,l}\,t}})+i\lambda_{F,l}\,{b^{2}_{F,l}}} \\
G^{F,l}(-y)&=&\log \left( exp(-\lambda_{F,l}t) \frac{(T^{F,l}_1(-y)+i \sqrt{i\lambda_{F,l}}\,b_{F,l})^2 }{(T^{F,l}_1(-y)+i T^{F,l}_2(t,-y))^2}\right)\\
\theta_l &=& tr(AMC_{l}A')= (a_{1l}-a_{2l})^2,\; l=1,2.\\
 T^{V,l}_1(-\theta_l y)&=& \sqrt {2\theta_l y-i\lambda_{V,l}\,{b^{2}_{V,l}}}\\
T^{V,l}_2(T,-\theta_l y)&=&\sqrt {-2\theta_l y(1-{{\rm e}^{-\lambda_{V,l}\,t}})+i\lambda_{V,l}\,{b^{2}_{V,l}}} \\
G^{V,l}(-\theta_l y)&=&\log \left( exp(-\lambda_{V,l}t) \frac{(T^{V,l}_1(-\theta_l y)+i \sqrt{i\lambda_{V,l}}\,b_{V,l})^2 }{(T^{V,l}_1(-\theta_l y)+i T^{V,l}_2(t,-\theta_l y))^2}\right)
 \end{eqnarray*}
 Which leads to:
 \begin{eqnarray*}
     I_F^{(l)}(\lambda_{F,l}T,-y)&=&{-\frac {2\,a_{F,l}}{\sqrt {i \lambda_{F,l}}}}\left[ -T^{F,l}_2(T,-y)+\sqrt{i \lambda_{F,l}}b_{F,l}+ \frac{i}{2}T^{F,l}_1(-y) G^{F,l}(-y) \right]+\lambda_{F,l} a_{F,l} b_{F,l} T\\
    &&\\ \nonumber
       I_V^{(l)}(\lambda_{V,l}T,- \theta_l y)&=&{-\frac {2\,a_{V,l}}{\sqrt {i \lambda_{V,l}}}}\left[ -T^{V,l}_2(T,- \theta_l y)+\sqrt{i \lambda_{V,l}}b_{V,l}+\frac{i}{2}T^{V,l}_1(- \theta_l y) G^{V,l}(-\theta_l y)\right]+\lambda_{V,l} a_{V,l} b_{V,l} T \\
       D I_F^{(l)}(\lambda_{F,l}T,-y)&=& {-\frac {a_{F,l}}{\sqrt {i \lambda_{F,l}}}} \int_0^{\lambda_{F,l} t} \frac{1-exp(-\lambda_{F,l} T+s)}{\sqrt{-2 y(1-exp(-\lambda_{F,l} t+s))+i \lambda_{F,l} b^2_{F,l} }}\;ds\\
       &=&  {\frac {a_{F,l}}{\sqrt {i \lambda_{F,l}}}} \int_{0}^{1-exp(-\lambda_{F,l} T)} \frac{v}{(v-1)\sqrt{-2 y v+i \lambda_{F,l} b^2_{F,l} }}\;dv\\
       D^n I_F^{(l)}(\lambda_{F,l}T,-y)&=&
       (-1)^{n} \prod_{k=2}^n (2k-3)\frac{a_{F,l}}{\sqrt{i \lambda_{F,l}} } \int_0^{-\lambda_{F,l} t} \frac{(1-exp(-\lambda_{F,l} T+s))^n}{(-2y (1-exp(-\lambda_{F,l} t+s))+i \lambda_{F,l} b^2_{F,l})^{\frac{2n-1}{2}}}\;ds\\
       &=&
       (-1)^{n+1} \prod_{k=2}^n (2k-3)\frac{a_{F,l}}{\sqrt{i \lambda_{F,l}} } \int_0^{1-exp(\lambda_{F,l} T)} \frac{v^n}{(v-1)(-2 y v+i \lambda_{F,l} b^2_{F,l} )^{\frac{2n-1}{2}}}\;dv\\
      D I_V^{(l)}(\lambda_{V,l}T,- \theta_l y)&=& {-\frac {a_{V,l}}{\sqrt {i \lambda_{V,l}}}}\theta^2_l \int_0^{-\lambda_{V,l} t} \frac{1-exp(-\lambda_{V,l} t+s)}{\sqrt{-2 \theta_l y(1-exp(-\lambda_{V,l} T+s))+i \lambda_{V,l} b^2_{V,l} }}\;ds\\
      &=&  {-\frac {a_{V,l}}{\sqrt {i \lambda_{V,l}}}}\theta^2_l \int_0^{1-exp(\lambda_{V,l} T)} \frac{v}{(v-1)\sqrt{-2 \theta_l y v+i \lambda_{V,l} b^2_{V,l} }}\;dv\\
     D^n I_V^{(l)}(\lambda_{V,l}T,- \theta_l y)&=&
       (-1)^n \prod_{k=2}^n (2k-3)\frac{a_{V,l}}{\sqrt{i \lambda_{V,l}} } \theta^{2n}_l\\
       && \int_0^{\lambda_{V,l} T} \frac{(1-exp(-\lambda_{V,l} T+s))^n}{(-2\theta_l y (1-exp(-\lambda_{V,l} T+s))+i \lambda_{V,l} b^2_{V,l})^{\frac{2n-1}{2}}}\;ds \\
        &=&
       (-1)^{n+1} \prod_{k=2}^n (2k-3)\frac{a_{V,l}}{\sqrt{i \lambda_{V,l}} } \theta^{2n}_l\\
       && \int_0^{1-exp(-\lambda_{V,l} T)} \frac{v^n}{(v-1)(-2 \theta_l y v+i \lambda_{V,l} b^2_{V,l} )^{\frac{2n-1}{2}}}\;dv
       \end{eqnarray*}
\begin{eqnarray*}
     K^+_1(-My)&=&\sum_{l=1}^2  I_F^{(l)}(\lambda_{F,l}T, -y)\\ \nonumber
   K^+_2(-My)&= & \sum_{l=1}^2   I_V^{(l)}(\lambda_{V,l}T,-\theta_l y)\\
      D^n K^+_1(-My)&=&\sum_{l=1}^2  D^n I_F^{(l)}(\lambda_{F,l}T, -y)\\ \nonumber
    D^n  K^+_2(-M y)&= & \sum_{l=1}^2    D^n  I_V^{(l)}(\lambda_{V,l}T,-\theta_l y)\\
\end{eqnarray*}

\end{document}